\documentclass[onecolumn,12pt]{IEEEtran}
\usepackage{amsmath,graphicx}
\usepackage{amssymb,amsfonts,amsthm}
\usepackage{setspace}
\usepackage{cite}
\usepackage{subcaption}
\usepackage{tikz,pgfplots}
\pgfplotsset{
      every axis/.append style={scale=0.45},
      every axis plot/.append style={mark repeat=2, mark size=3pt,very thick},
      label style={font=\small},
      legend style={font=\footnotesize}
      }
\usetikzlibrary{external,matrix,positioning,graphs}
\newtheorem{theorem}{Theorem}
\newtheorem{proposition}{Proposition}
\newtheorem{definition}{Definition}
\newtheorem{corollary}[theorem]{Corollary}

\newtheorem{lemma}[theorem]{Lemma}

\newcommand{\mc}{\mathcal}

\newcommand{\norm}[1]{\left\|{#1}\right\|}

\begin{document}

\title{Rate-Distortion Bounds on Bayes Risk in Supervised Learning}

\author{\IEEEauthorblockN{Matthew Nokleby\IEEEauthorrefmark{1}, Ahmad Beirami\IEEEauthorrefmark{2}, Robert Calderbank\IEEEauthorrefmark{3}} \\
       \IEEEauthorblockA{\IEEEauthorrefmark{1}Wayne State University, Detroit, MI, email: matthew.nokleby@wayne.edu \\
                         \IEEEauthorrefmark{2}MIT, Cambridge, MA, email: beirami@mit.edu\\
                         \IEEEauthorrefmark{3}Duke University, Durham, NC, email: robert.calderbank@duke.edu}
       \footnote{This work was presented in part at the IEEE Machine Learning in Signal Processing Workshop, Boston, MA, Oct. 2015, and the IEEE Symposium on Information Theory, Barcelona, Spain, Jul. 2016.}
       }

\maketitle

\onehalfspacing

\begin{abstract}
We present an information-theoretic framework for bounding the number of labeled samples needed to train a classifier in a parametric Bayesian setting. Using ideas from rate-distortion theory, we derive bounds on the average $L_p$ distance between the learned classifier and the true maximum a posteriori classifier---which are well-established surrogates for the excess classification error due to imperfect learning. We provide lower and upper bounds on the rate-distortion function, using $L_p$ loss as the distortion measure, of a maximum {\em a priori} classifier in terms of the differential entropy of the posterior distribution and a quantity called the interpolation dimension, which characterizes the complexity of the parametric distribution family. In addition to expressing the information content of a classifier in terms of lossy compression, the rate-distortion function also expresses the minimum number of bits a learning machine needs to extract from training data in order to learn a classifier to within a specified $L_p$ tolerance. Then, we use results from universal source coding to express the information content in the training data in terms of the Fisher information of the parametric family and the number of training samples available. The result is a framework for computing lower bounds on the Bayes $L_p$ risk. 
This framework complements the well-known probably approximately correct (PAC) framework, which provides minimax risk bounds involving the Vapnik-Chervonenkis dimension or Rademacher complexity. Whereas the PAC framework provides upper bounds the risk for the worst-case data distribution, the proposed rate-distortion framework lower bounds the risk averaged over the data distribution. We evaluate the bounds for a variety of data models, including categorical, multinomial, and Gaussian models. In each case the bounds are provably tight orderwise, and in two cases we prove that the bounds are tight up to multiplicative constants.
\end{abstract}


% !TEX root = it.tex

\section{Introduction}

A central problem in statistics and machine learning is {\em supervised learning}, in which a learning machine must choose a classifier using a sequence of labeled training samples drawn from an unknown distribution. The effectiveness of the learned classifier is measured by its accuracy in classifying future test samples drawn from the same distribution. Standard approaches to this problem include support vector machines, \cite{vapnik:98,vapnik:00,muller:NN01}, random forests \cite{breiman:ML01}, and deep neural networks \cite{hinton:NC06,bengio:NIPS07}.

In supervised learning, a fundamental question is the {\em sample complexity}: how many training samples are necessary to learn an effective classifier? The prevailing approach to characterizing the sample complexity is the {\em probably approximately correct} (PAC) framework, which provides almost sure bounds on the sample complexity of families of classifiers irrespective of the data distribution. These bounds are expressed in terms of the Vapnik-Chervonenkis (VC) dimension, which captures combinatorially the complexity of families of classifiers \cite{vapnik:TPA71,vapnik:NN99}. A typical result goes as follows: for a classifier family with VC dimension $h$ and given $n$ training samples, the excess error probability of the learned classifier over that of the best classifier in the family is with high probability $O(\sqrt{h/n})$. The PAC framework leads to the empirical risk minimization (ERM) and structural risk minimization (SRM) frameworks for model selection: The system designer considers sequence of classifier families with increasing VC dimension, and chooses the family that minimizes the PAC bound over the available training set. PAC bounds are available for many popular classifiers, including SVMs and neural networks \cite{vapnik:98,shawetaylor:IT98,Barron-approximation-bounds}. Refinements to the PAC bounds provide tighter bounds on the risk. Data-dependent bounds based on the fat-shattering dimension and Rademacher complexity were developed in \cite{bartlett:JCSS96,alon:JACM97,bartlett:JMLR03}, and more recently, {\em local} Rademacher averages, margin-dependent, and concentration-free bounds tighten the results further, in some cases offering an order-wise improvement in predicted sample complexity \cite{bartlett:AS05,massart:AS06,mendelson:JACM15}. PAC-Bayes bounds, in which one imposes a prior over the set of classifiers, were developed in \cite{mcallester:COLT99,seeger:JMLR03}.

A limitation of PAC bounds is that they characterize the minimax performance over the distribution family. This may lead to pessimistic predictions relative to practical peformance \cite{cohn:NC92,haussler:ML96}. Indeed, the authors of \cite{chen:arxiv16} put it this way: ``The assumption that the adversary is capable of choosing a worst-case parameter is sometimes over-pessimistic. In practice, the parameter that incurs a worst-case risk may appear with very small probability.'' 
To go around this, many researchers have studied {\em average-case} bounds. Rissanen proposed the minimum description length (MDL) criterion for model selection \cite{rissanen:A78}, which leverages results from universal source coding to select the complexity of the model class and to avoid overfitting. The MDL framework has since seen wide use in machine learning (see~\cite{barron-rissanen-yu} and~\cite{shamir-online-learning} and the references therein for a recent survey). Information-theoretic connections to model-order selection have also been studied, resulting in the Akaike and Bayes information criteria \cite{akaike:ISIT73,burnham:SMR04} and information bottleneck methods \cite{tishby:arxiv2000,shamir:ALT08}

In this paper, we develop a framework for computing bounds on the Bayes risk for estimating the posterior in supervised learning. In particular, we are concerned with the ``soft'' classification performance of the learning machine. That is, rather than measure performance strictly in terms of the error probability of the learned classifier, we measure how well a learning machine can estimate the {\em posterior} function, which in turn is used to classify test samples via the MAP rule. The quality of the estimated posterior measures not only how well one can classify future test samples, but also how well one can estimate the confidence level of the learned classifier each time it encounters a test point.
We develop the framework of this paper in a Bayesian parametric setting. The joint distribution on data points $X$ and labels $Y$ belongs to a known parametric family $p(x,y|\theta)$, and the parameters that index the distribution are drawn from a known prior $q(\theta)$. An example is Gaussian classification, where for each class $y$, $p(x|y,\theta)$ is a multivariate Gaussian with fixed covariance and mean taken as a subvector of $\theta$. A suitable prior $q(\theta)$ for computational purposes is the conjugate prior, which in this case is itself a Gaussian.

The proposed framework provides lower bounds on the average $L_p$ distance between the true posterior $p(y|x,\theta)$ and the posterior estimated from $n$ i.i.d. samples drawn from $p(x,y|\theta)$. Because the bounds are averaged over the prior $q(\theta)$, they do not exhibit the pessimism of minimax bounds. The $L_p$ risk is a well-known surrogate for the excess classification error \cite{devroye:85,ney:PRIA03}, so bounds on these errors give insight into the performance of the learned classifier. Furthermore, this approach connects the problem of learning a classifier to the problem of learning a distribution from samples---for a fixed $x$, the posterior is merely a distribution to learn from training samples. The problem of learning a distribution has a rich history, dating back to the Good-Turing estimator \cite{good:Biometrika53} and continuing to recent results \cite{drmota:IT04,kamath:15,han:IT15,valiant:15}.

The proposed framework identifies a relationship between supervised learning and lossy source coding. In the parametric Bayesian setting, the posterior distribution is a function of the random parameters $\theta$ and therefore is a random object. If we take the $L_p$ distance as the distortion function, we can bound the number of nats needed to describe the posterior to within a specified tolerance.\footnote{Note that information and entropy are measured in nats throughout this paper.} What follows is the main result of this paper: In order to drive the average $L_p$ error below a threshold $\epsilon$, the mutual information between the training samples and the parameters $\theta$ must be at least as great as the differential entropy of the posterior plus a penalty term that depends on $\epsilon$ and a sample-theoretic quantity, called the {\em interpolation dimension}, which measures the number of data points from the posterior distribution needed to uniquely interpolate the entire function.

The resulting framework is complementary to the PAC framework. Whereas the PAC framework considers families of classifiers and provides generalization bounds that hold for any data distribution, the rate-distortion framework considers families of data distributions and provides generalization bounds that hold for any classifier. Whereas the VC dimension characterizes the combinatorial complexity of a family of classifiers, the interpolation dimension characterizes the sample complexity of a parametric family of data distributions. The larger the interpolation dimension, the more training samples are needed to guarantee classifier performance. We also emphasize that Bayes risk lower bounds are derived in terms of $f$-divergences, which generalize the usual KL-divergence, in \cite{chen:arxiv16}. An explicit connection between rate-distortion theory and learning is investigated in \cite{raginsky:ITW07}, where PAC-style bounds on generalization error are derived when samples are subject to lossy compression. Along similar lines, lower bounds on the {\em distributed} learning of a function over a network are derived in \cite{xu:IT17}.

The contributions of this paper are as follows. After formally laying out the problem statement in Section \ref{sect:preliminaries}, we present bounds on the rate-distortion functions of Bayes classifiers in Section \ref{sect:main.results}. These rate-distortion functions take the posterior $p(y|x,\theta)$ as the random source to be represented, and the $L_p$ risk as the distortion measure. We consider two definitions of the $L_p$ Bayes risk: one which averages the $L_p$ risk over the parameter $\theta$ {\em and} the test point $X$, and one which averages over the parameter $\theta$ but considers the {\em worst case} $L_p$ over test points $X$ that live in a pre-defined subset. The first definition characterizes the average performance overall, whereas the second definition allows one to focus on a particular region of test points.

The bounds on the rate-distortion function are in terms of an object called an {\em interpolation set} and a related quantity called the {\em interpolation dimension}. It is difficult to directly analyze the information content in the posterior of a continuous distribution, as it is a set of uncountably many random vectors. To address this issue, we define a sufficient set of points $x$ to describe the posterior, called a (sufficient) interpolation set. When the interpolation set has finite cardinality, one can more easily bound the rate-distortion function by considering only the finite samples of the posterior. The resulting bounds involve the differential entropy of the posterior evaluated at the elements of the interopolation set, the $L_p$ distortion criterion, and the cardinality of the interpolation set, termed the {\em interpolation dimension}.

In Section \ref{sect:sample.complexity}, we translate the bounds on the rate-distortion bounds into bounds on the sample complexity. Applying a Bayesian version of the capacity-redundancy theorem, \cite{Clarke_Barron}, we find that the mutual information between the training set and the parameters $\theta$, which are random in our Bayesian setup, scales as $\log(n)$ and depends on the determinant of the Fisher information matrix averaged over the distribution family. Using this fact, we derive bounds on the number of samples needed to ensure small $L_p$ error. We also discuss the challenges and opportunities for deriving matching Bayes risk outer bounds.

In Section \ref{sect:numerical.examples}, we consider several case studies. First, we consider the simple problem of estimating a discrete distribution from $n$ i.i.d. samples. In this case, derive closed-form bounds on the $L_p$ rate-distortion function of the distribution and lower bounds on the $L_p$ Bayes risk. These bounds are tight order-wise, and in the asymptote they agree almost exactly with the minimax error. Then, we consider learning a binary {\em multinomial} classifier, a model popular in document classification. Again we derive closed-form bounds on the rate-distortion function and $L_p$ Bayes risk, which are provably order optimum. We carry out a similar analysis for binary Gaussian classifiers. Finally, we consider a simple ``zero-error'' classification problem. In this case, the resulting $L_p$ risk falls off as $1/n$ instead of the $\sqrt{1/n}$ obtained in the previous cases; we also show that the rate-distortion bounds are nearly tight.

We give our conclusions in Section \ref{sect:conclusion}.

{\bf Notation:} Let $\mathbb{R}$ and $\mathbb{Z}$ denote the fields of real numbers and integers, respectively, and let $\mathbb{R}_+$ denote the set of non-negative reals.. Let capital letters $X$ denote random variables and vectors, and lowercase letters $x$ denote their realizations. For $x \in \mathbb{R}^k$, let $\mathrm{diag}(x)$ denote the $k \times k$ matrix with the elements of $x$ on its diagonal. We let $E[\cdot]$ denote the expectation, with the subscript indicating the random variable over which the expectation is taken when necessary. Let $|\cdot|$ denote the cardinality of a set. For a function $f(x)$ and a finite set $S$, let $\{f(x)\}_S$ denote the $|S|$-length vector of function evaluations of $f$ at the points in $S$, suppressing the arguments when clarity permits. Let $[M] = \{1,\dots,M\}$ for integer $M$. Let $I(\cdot \  ;\cdot)$ denote the mutual information and $h(\cdot)$ denote the differential entropy. 
We use the natural logarithm throughout, so these quantities are measured in nats. 
Let $\Delta_{k}$ denote the $k$-dimensional unit simplex
\begin{equation*}
	\Delta_k = \left\{ x \in \mathbb{R}_+^{k+1} : \sum_{i=1}^{k+1} x_i = 1 \right\}.
\end{equation*}
For $z \in \mathbb{R}_+$, let $\Gamma(z)$ denote the gamma function:
\begin{equation*}
	\Gamma(z) = \int_0^\infty x^{z-1}e^{-x}dx. 
\end{equation*}
For $z \in \mathbb{R}_+$, let $\psi(z)$ denote the digamma function:
\begin{equation*}
	\psi(z) = \frac{d}{dz}\log \Gamma(z).
\end{equation*}
For two scalars $x,y \in \mathbb{R}_+$, let $B(x,y)$ denote the beta function:
\begin{equation*}
	B(x,y) = \frac{\Gamma(x)\Gamma(y)}{\Gamma(x+y)}.
\end{equation*}
For a vector $\gamma \in \mathbb{R}_+^M$, let $B(\gamma)$ denote the {\em multivariate} Beta function:
\begin{equation*}
	B(\gamma) = \frac{\prod_{i=1}^M\Gamma(\gamma_i)}{\Gamma\left(\sum_{i=1}^M \gamma_i\right)}.
\end{equation*}
Let $\mathcal{N}(\mu,\Sigma)$ denote the normal distribution with mean $\mu$ and (co)variance $\Sigma$. Let $\mathrm{Beta}(\alpha,\beta)$ denote the beta distribution with shape parameters $\alpha,\beta > 0$, which has the density function
\begin{equation*}
	p(x) = \frac{x^{\alpha-1}(1-x)^{\beta-1}}{B(\alpha,\beta)}.
\end{equation*}
For $\gamma \in \mathbb{R}_+^k$, let $\mathrm{Dir}(\gamma)$ denote the Dirichlet distribution, which has the density function
\begin{equation*}
      p(x) = \frac{1}{B(\gamma)}\prod_{i=1}^M x_i^{\gamma_i - 1}.
\end{equation*}

% !TEX root = it.tex

\section{Problem Statement}\label{sect:preliminaries}
We consider the problem of supervised learning in a parametric statistical framework. Let each data point $X \in \mathcal{X} \subset \mathbb{R}^d$ and its label $Y \in [M]$ be distributed according to $p(x,y | \theta)$, where $\theta \in \Lambda \subset \mathbb{R}^k$ indexes a parametric family of distributions $\mathcal{D} = \{p(x,y|\theta) : \theta \in \Lambda\}$. The alphabet $\mathcal{X}$ may be discrete or continuous. In the former case, $p(x,y|\theta)$ is the joint probability mass function of the data point and its label. In the latter case, we abuse notation slightly and refer to $p(x,y|\theta)$ as the joint probability {\em density} function even though $Y$ is a discrete random variable. 

Suppose that ``nature'' selects $\theta \in \Lambda$. The learning machine obtains a sequence of $n$ samples, denoted $Z^n = (X^n,Y^n)$, where each pair $Z_i = (X_i,Y_i), \ 1 \leq i \leq n$ is drawn i.i.d. according to $p(x,y|\theta)$. The learning task is to select a classifier $\hat{y} = w(x)$ from the training samples $Z^n$. In principle, the classifier may be any function $w : \mathcal{X} \to [M]$. If $\theta$ were known, one could choose the the maximum {\em a posteriori} (MAP) classifier, which minimizes the classification error:
\begin{equation*}
	w_{\mathrm{MAP}}(x) = \arg\max_y p(y|x,\theta),
\end{equation*}
where $p(y|x,\theta)$ is calculated according to Bayes' rule. Of course, in supervised learning the data distribution is unknown, so the MAP classifier is unavailable. Instead, we suppose that the learning machine knows the parametric family $\mathcal{D}$, but not the specific distribution $p(x,y;\theta)$. 

The objective of supervised learning is to characterize the performance of the learned classifier $w(x)$ as a function of the number of training samples $n$. A natural performance metric is the gap between the misclassification error of the learned classifier and that of $w_\mathrm{MAP}$:
\begin{equation}
	L_c(x,\theta;w,w_\mathrm{MAP}) = \mathrm{Pr}(Y \neq w(x)) - \mathrm{Pr}(Y \neq w_{\mathrm{MAP}}(x)),
	\label{eq:classification-loss}
\end{equation}
where the probabilities are computed according to the joint distribution $p(x,y|\theta)$. As discussed in the introduction, the minimax loss with respect to $L_c$ is characterized by the PAC framework. For a family of classifiers containing the MAP classifier and having with VC dimension $h$, $L_c(x,\theta;w,w_\mathrm{MAP}) = O(\sqrt{h/n})$ for any distribution over $X$ and $Y$ and with high probability over the distribution $p(x,y;\theta)$.

Instead of the misclassification error gap, we analyze the {\em Bayes risk} in learning the posterior to analyze the {\em soft} classification capability of the learning machine, where performance is averaged over the distributions indexed by $\theta$. Let $q(\theta)$ be a prior distribution over the parametric family. The proposed framework presents performance bounds averaged over the family of distributions according to $q(\theta)$. We can view the Bayes error in a few different ways. First, if $q(\theta)$ represents the true distribution over the parameters space, then the Bayes risk is simply the average loss {over many instances of the learning problem. Second, for any $q(\theta)$, the Bayes risk represents a lower bound on the minimax risk, and depending on the strength of the prior distribution $q(\theta)$ the Bayes risk may be much smaller. 

Furthermore, rather than study the classification error gap, we study the $L_p$ loss. We define these losses in terms of the {\em posterior distribution} $p(y|x,\theta)$, also called the {\em regression function}. That is, rather than choose a classifier $w(x)$ directly, the learning machine estimates the regression function, which is later used to classify test points according to the MAP rule. To underscore this point, let
\begin{equation*}
	W(y|x,\theta) := p(y|x,\theta)
\end{equation*}
denote the true regression function, which takes as input $x \in \mathcal{X}$ and produces as output the $M$-dimensional vector of probabilities that $x$ belongs to class $y$. Also let $\delta(Z^n) = \hat{W}(y|x)$ be a {\em learning rule} that maps the training samples 
to an estimate\footnote{We omit any dependence on $\theta$ in $\hat{W}(y|x)$ to indicate that the regression function estimate is made in ignorance of $\theta$.} of the regression function $W$. Then, for every $x$, the loss is defined as the $L_p$ distance between the $M$-dimensional vector formed by $W(\cdot|x;\theta)$ and $\hat{W}(\cdot|x)$:
\begin{equation}
	L_p(x,\theta;W,\hat{W}) = \left(\sum_{y=1}^M |W(y|x,\theta) - \hat{W}(y|x)|^p\right)^{\frac{1}{p}} \label{eq:L-inf-define}.
\end{equation}

Bounds on the $L_p$ loss, rather than the classification error $L_c$, are valuable for several reasons. First, The $L_p$ loss is related to $L_c$.
A well-known fact (see, e.g. \cite{devroye:85}) is that $L_1$, averaged over $\mathcal{X}$, bounds the classification loss $L_c$ from above. A somewhat less well-known fact is that this relationship holds pointwise (see~\cite{ney:PRIA03} for a discussion). A classifier using the regression function $\hat{W}(y|x)$ has classification error satisfying $L_c \leq L_1$ both for any point $x$ and averaged over $X$. Via norm equivalence, one can derive similar bounds for any $L_p$.
Second, the $L_p$ loss gives a comprehensive sense of classifier performance. When classifying signals, one wants to know not only the most probable class, but the confidence one has in the classification. Small $L_p$ risk not ensures not only near-optimum classification error, but also good estimates on the accuracy of each classification decision. Similarly, if one wants to use a classifier to minimize an arbitrary expected loss---in which each type of misclassification incurs a different penalty---one needs an accurate estimate of the regression function itself.

Finally, we point out the well-known relationship between $L_1$ and the Kullbeck-Leibler (KL) divergence. Let
\begin{equation}
	L_{\mathrm{KL}}(x,\theta;W,\hat{W}) = \sum_{y=1}^M W(y|x,\theta) \log\left(\frac{W(y|x,\theta)}{\hat{W}(y|x)} \right)
\end{equation}
be the KL divergence between the $W$ and $\hat{W}$, evaluated at each test point $x \in \mathcal{X}$. Pinsker's inequality \cite{pinsker:IAN60,csiszar:SSMH67} states that $L_{\mathrm{KL}} \geq \frac{L_1^2}{2}$. Therefore, lower bounds on the $L_1$ Bayes risk can be translated into bounds on the KL divergence between $W$ and $\hat{W}$, averaged over $\theta$. The KL divergence has an important place in modern machine learning in the guise of the {\em cross-entropy} loss or log loss, which is a popular criterion for the training of learning models, including deep learning \cite{Goodfellow.etal.2016}.

The Bayes risk of a learning rule $\delta$ is computed by taking the average over the point-wise loss functions defined above. There are three variables over which to take averages: the data point $X$, the training set $Z^n$, and the parameterization index $\theta$. It is illustrative to parse out the impact of averaging over different random variables. To this end, we consider two definitions of the Bayes risk.

The first definition, termed the $\mathcal{X}^\prime$-Bayes risk, involves averages over $\theta$ and $Z^n$ only. We suppose that the test point $x$ lives in a set $\mathcal{X}^\prime \subset \mathcal{X}$, over which set we consider the {\em worst-case} $L_c$ or $L_p$ loss. If $\mathcal{X}$ is compact, we may consider the worst-case Bayes risk for all $x \in \mathcal{X}$; otherwise, it may be beneficial to consider the worst-case performance of some compact subset $\mathcal{X}$. That is, we consider the $L_c$ or $L_p$ loss averaged over $\theta$ and $Z^n$, in the worst case over $x \in \mathcal{X}^\prime$. We formalize this with the following definition.
\begin{definition}\label{def:gamma.bayes.risk}
	Define the {\em $\mathcal{X}^\prime$-Bayes risk} of a learning rule $\delta$ with respect to the $L_p$ loss as
	\begin{equation}\label{eqn:gamma.bayes.risk}
		\mathcal{L}_p^{\mathcal{X}^\prime}(\delta) = \sup_{x \in \mathcal{X}^\prime} \left(E_{Z^n, \theta}\left[\sum_{y=1}^M |W(y|x;\theta) - \hat{W}(y|x)|^p\right]\right)^{\frac{1}{p}}.
	\end{equation}
\end{definition}
An analogous definition holds with respect to the $L_c$ loss.

The worst-case Bayes risk over points $x \in \mathcal{X}^\prime$ may be pessimistic, especially if the set $\mathcal{X}^\prime$ is large. We also consider the $L_p$ loss averaged over $\theta$, $Z^n$, and $X$, which we term simply the {\em Bayes risk}.
\begin{definition}\label{def:bayes.risk}
	Define the {\em Bayes risk} with respect to the $L_p$ loss as
	\begin{equation}
       \mathcal{L}_p(\delta) = \left(E_{X,Z^n,\theta}\left[\sum_{y=1}^M |W(y|x;\theta) - \hat{W}(y|x)|^p\right]\right)^{\frac{1}{p}}.
    \end{equation}
\end{definition}
Again, an analogous definition holds with respect to the $L_c$ loss.
The Bayes risk $\mathcal{L}$ is simply the average performance, measured in terms of $L_c$ or $L_p$ loss, averaged over the data distribution $\theta$, the data point $X$ and the training set $Z^n$. Note that in each case the normalizing power $1/p$ is taken outside the expectation.

The basic ingredient of our results is a {\em rate-distortion} analysis of the Bayes risk. In essence, we characterize the minimum number of nats that the learning machine must extract from the training set in order to obtain a representation of the regression function up to a specified Bayes risk tolerance. To this end, we characterize the rate-distortion function of the posterior that the learning machine hopes to learn from $Z^n$. Suppose $\hat{W}$ is the ``compressed'' or lossy version of the posterior $W$. Then, define the rate distortion functions with respect to the Bayes risk functions from Definitions \ref{def:gamma.bayes.risk} and \ref{def:bayes.risk}.

\begin{definition}
	The {\em rate-distortion function} of the regression function $W$ with respect to the $\mathcal{X}^\prime$-Bayes risk and the $L_p$ loss is
	\begin{equation}
		R_p^{\mathcal{X}^\prime}(D) = \inf_{\substack{p(\hat{W} | W) \\ \mathcal{L}_p^{\mathcal{X}^\prime} \leq D}} I(W;\hat{W}),
	\end{equation}
	where $I(W;\hat{W})$ is the mutual information between the true and approximated regression function. With respect to the Bayes risk, the rate-distortion function is
	\begin{equation}
		R_p(D) = \inf_{\substack{p(\hat{W} | W) \\ \mathcal{L}_p \leq D}} I(W;\hat{W}).
	\end{equation}
\end{definition}
The challenge in computing the rate-distortion functions defined above is that the regression function is a collection of many random variables---uncountably many if $\mathcal{X}$ is an uncountable alphabet. While the mutual information $I(W;\hat{W})$ between them is well-defined, analyzing $I(W;\hat{W})$ requires care. Much of Section \ref{sect:main.results} is given over to the development of techniques for such analysis.

The rate-distortion function is interesting in its own right in the usual information-theoretic sense. Indeed, if one has learned a regression function for a parametric model, one might ask how much information is required to encode the posterior to transmit to another party. Per the rate-distortion and source-channel separation theorems, one needs $R_p^{\mathcal{X}^\prime}(D)$ or $R_p(D)$ nats in order to ensure that the reconstructed posterior has $\mathcal{X}^\prime$-Bayes or Bayes risk, respectively, no more than $D$. Therefore, the following rate-distortion analysis has implications for {\em distributed} learning over networks, which is a subject to be taken up in future work.

Nevertheless, our main motivation in studying the rate-distortion function of $W$ is to derive lower bounds on the Bayes risk. In addition to quantifying how many nats one needs to encode the regression function up to a Bayes risk tolerance $D$, the rate-distortion function quantifies how many nats a learning machine needs to extract from the training set in order to learn the regression function up to the same tolerance. Furthermore, one can quantify the maximum number of nats one can extract from the training set via the mutual information between $Z^n$ and the distribution index $\theta$. Putting the two ideas together, one can derive necessary conditions on the Bayes risk. We formalize this notion in the following lemma.
\begin{lemma}\label{lem:rd.to.bayes.risk}
	Whenever a learning rule $\delta(Z^n)$ has $\mathcal{X}^\prime$-Bayes risk or Bayes risk less than or equal to $D$, the conditions
	\begin{align}
		I(Z^n; \theta) &\geq R_p^{\mathcal{X}^\prime}(D) \\
		I(Z^n; \theta) &\geq R_p(D)
	\end{align}
	hold, respectively.
\end{lemma}
	See Appendix \ref{app:proof-lemmas} for the proof.
The intuition behind Lemma \ref{lem:rd.to.bayes.risk} is the number of nats {\em required} to learn $W$ with Bayes risk no greater than $D$ is given by $R_p(D)$, and the number of nats {\em provided} by the training set $Z^n$ is no greater than the mutual information $I(Z^n; \theta)$. The number of nats provided must satisfy the number of nats required. We further illustrate the analogy between the proposed framework and standard rate-distortion theory in Figure \ref{fig:rate.distortion.framework}.

\begin{figure}[t]
      \centering
      \begin{tikzpicture}%[text height=1.5ex,text depth=.25ex]
	
	\matrix [column sep=7mm, row sep=3mm]
	{
		\node (distribution) {$p(x)$}; &
		\node (source) {$X^n = (X_1,\dots,X_n)$}; &
		\node (encoder) [shape=rectangle,draw] {Encoder}; &
		\node (code) {$i \in \{1,\dots,2^{nR}\}$}; &
		\node (decoder) [shape=rectangle,draw,text width=0.6in,align=center] {Decoder}; &
		\node (reconstruction) {$\hat{X}^n$}; \\
		
		\node (prior) {$q(\theta)$}; &
		\node (posterior) {
			$\begin{aligned}
				p(x,y;\theta) \\
				W(y|x;\theta)
			\end{aligned}$}; & 
			& 
			\node (training) {$(Z_1,\dots,Z_n)$}; & 
			\node (learner) [shape=rectangle,draw,align=center,text width=0.6in] {Learning Machine}; &
			\node (estimate) {$\hat{W}(y|x)$}; \\
	};

	\graph {
		(distribution) -> (source) -> (encoder) -> (code) -> (decoder) -> (reconstruction);
	};
	\graph {
		(prior) -> (posterior) -> (training) -> (learner) -> (estimate);
	};

\end{tikzpicture}
      \caption{The analytical framework in this paper is by a connection to rate-distortion theory. In rate distortion, a source distribution $p(x)$ gives rise to an $n$-length sequence $X^n$, which is encoded to one of $2^{nR}$ indices. The decoder infers from this index a noisy reconstruction $X^n$, and the average distortion depends on the encoding rate via the rate-distortion function. In this paper, the prior distribution $q(\theta)$ gives rise to the data distribution $p(x,y;\theta)$ and its associated regression function $W(y|x,\theta)$, and we treat the training samples $Z^n$ drawn from $p(x,y|\theta)$ as an imperfect encoding of the regression function. The learning machine infers from $Z^n$ a noisy estimate $\hat{W}(y|x)$, and the $L_p$ estimation error depends on the number of samples. \label{fig:rate.distortion.framework} }
\end{figure}
% !TEX root = it.tex

\section{Main Results}\label{sect:main.results}
This section is devoted to developing bounds on the rate-distortion functions, and by extension the Bayes risk functions, defined in the previous section.

We first define a few necessary concepts. The regression function $W(y|x,\theta)$ is a potentially uncountable collection of random variables, one for each point $(x,y) \in \mathcal{X} \times [M]$. The mutual information between, or the joint entropy of, uncountably many random variables is difficult to analyze directly, which makes it difficult to compute the rate-distortion functions of $W$. Therefore, we will analyze the information-theoretic quantities of a {\em sampled} version of $W(y|x;\theta)$, which acts as a sufficient statistic for the entire function. We capture this notion by defining the {\em interpolation set} and the {\em interpolation dimension}.

\begin{definition}
Let $S \subset \mathcal{X}$ be a finite set, and let $W(S)$ be the $M-1 \times |S|$ matrix
\begin{equation}\label{eqn:interpolation.set.definition}
	W(S) = 
	\begin{bmatrix}
		W(1|x_1;\theta) & \dots &  W(1| x_{|S|}) \\
		\vdots & \ddots & \ddots \\
		W(M-1|x_1;\theta) & \dots &  W(M-1| x_{|S|})
	\end{bmatrix}.
\end{equation}
That is, $W(S)$ is a matrix where the columns are evaluations of the first $M-1$ points of the regression function at the points in $S$. We say that $S$ is an {\em interpolation set} for the regression function $W(y|x;\theta)$ if the differential entropy $h(W(S))$ is well-defined and finite.
\end{definition}
In other words, $S$ is an interpolation set if sampling the regression function at each point $x \in S$ does not over-determine the regression function with respect to the randomness in $q(\theta)$. For example, in Section \ref{sect:numerical.examples} we consider a binary classification problem over $\mathcal{X} = \mathbb{R}^d$ where the regression function is has the logistic form:
\begin{equation*}
  W(y=1|x,\theta) = \frac{1}{1+\exp(-x^T \theta)}.
\end{equation*}
Recall that $\theta$ is the unknown parameter, here playing the role of the regression coefficients. If one chooses $S$ to be a set of $d$ linearly independent vectors $x_1,\dots,x_d$ in $\mathbb{R}^d$, it is straightforward to verify that the joint density on the random variables $W(S)$ exists, and the joint differential entropy $h(W(S))$ is finite as long as $q(\theta)$ is well-behaved. However, if one adds another point to $S$, the first $d$ points completely determine the regression function value at the $(d+1)$-th point; the resulting joint distribution is singular, and the joint entropy is, depending on one's definitions, undefined or equal to $-\infty$.

For an interpolation set $S$, $W(S)$ provides a finite-dimensional representation of the (perhaps) infinite-dimensional regression function. Even when $\mathcal{X}$ is discrete, $W(S)$ gives a compact representation of $W$. It follows from the data-processing inequality that, for any learning rule,
\begin{equation}
	I(W(S);\hat{W}(S)) \leq I(W;\hat{W}).
\end{equation}
An important special case is when this inequality holds with equality.
\begin{definition}
	An interpolation set $S$ is said to be {\em sufficient} if $I(W(S);\hat{W}(S)) = I(W;\hat{W})$. Equivalently, an interpolation set is sufficient if $I(W;\hat{W}|W(S)) =I(W;\hat{W}|\hat{W}(S)) = 0$.
\end{definition}
Roughly speaking, an interpolation set is sufficient if one can recover the entire regression function $W(y|x;\theta)$ from the samples $W(S)$. Indeed, in the logistic regression example considered above, a set of $d$ linearly independent points for $S$ is a sufficient interpolation set. From the $d$ function evaluations, one can solve for $\theta$ exactly and recover the regression function for all values of $x$.

The cardinality of a sufficient interpolation set will play a prominent role in the analysis.
\begin{definition}
	The {\em interpolation dimension}, denoted $d_I$ is the cardinality of the smallest sufficient interpolation set $S$.
\end{definition}
The interpolation dimension characterizes the number of distinct evaluations of the regression function are needed to reconstruct it. In this sense, it is akin to the Nyquist rate in sampling theory, expressing how many function evaluations it takes to characterize a function having known structure. The interpolation dimension is a characteristic of the parametric family $\mathcal{D}$. Indeed, it measures the complexity of $\mathcal{D}$ in a manner similar to the VC dimension. Whereas the VC dimension characterizes the complexity of a family of classifiers by how many points it can shatter, the interpolation dimension characterizes the complexity of a family of {\em distributions} by how many sample points of the regression function are needed to reconstruct it.

We emphasize that the number of function evaluations of the regression function needed for an interpolation set is distinct from the number of independent samples drawn from the distribution needed to {\em learn} the regression function. The learning machine never sees regression function evaluations during supervised learning, only samples $Z^n$ drawn from the source distribution. The interpolation set and dimension are tools to facilitate the analysis of the rate-distortion functions and the Bayes risk associated with a parametric family. Nevertheless, we will show that, although they are distinct concepts, the number of training samples needed to learn the regression function is {\em related} to the interpolation dimension.

Before presenting the bounds, we need a final technical condition on $S$ and the parametric family.
\begin{definition}
  An interpolation set $S$ is said to be {\em onto} if, for every matrix $Q \in M-1 \times |S|$ in the set
  \begin{equation*}
    \mathcal{Q} = \{Q \in M-1 \times |S| : \sum_{i=1}^{M-1}q_{ij} \leq 1 \forall j,\  q_{ij} \geq 0, \forall i,j\},
  \end{equation*}
  there is at least one $\theta \in \Lambda$ such that $W(S) = Q$.
\end{definition}
In other words, an interpolation set $S$ defines a mapping $W(S): \Lambda \to \mathcal{Q}$, and the $S$ is onto if this mapping is onto. The set $\mathcal{Q}$ is just the set of all valid probability vectors truncated to their first $M-1$ elements. An interpolation set is onto if we can realize any valid probability vector by choosing $\theta$. For all of the examples we consider in this paper, the interpolation sets are onto. However, one can define parametric models where the regression function takes on only a subset of all possible probability vectors. Consider the trivial example $\mathcal{X} = \{0\}$, $M=2$, $\Lambda = [0,0.5)$, and the Bernoulli parametric family
\begin{equation*}
  \mathcal{D} = \left\{p(x,y;\theta) = \theta^{y-1} : \theta \in \Lambda \right\}.
\end{equation*}
In this case, the set of regression functions is restricted, and no onto interpolation set exists.

\subsection{Bounds involving $\mathcal{L}_p^{\mathcal{X}^\prime}(\delta)$}
The first result is a bound on the rate-distortion function $R^{\mathcal{X}^\prime}_p(D)$ in terms of the interpolation dimension and the entropy of the regression function.
\begin{theorem}\label{thm:pointwise.l1}
	Let $d_I(W)<\infty$ be the interpolation dimension of $W(y|x;\theta)$, and suppose there exists an onto, sufficient interpolation set with cardinality $d_I$. Let $S \subset \mathcal{X}^\prime$ be an interpolation set (not necessarily sufficient or onto) with $|S|=d_*$ such that $S \subset \mathcal{X}^\prime$. Then, the rate-distortion function $R_p^{\mathcal{X}^\prime}(D)$ for all $p\geq 1$ is bounded by
	\begin{align}\label{eq:pointwise.lp}
    	R_p^{\mathcal{X}^\prime}(D) &\geq \left[ h(W(S)) - d_* (M-1) \left( \log D +  \log \left(2 \Gamma \left(1+\frac{1}{p} \right)\right) + \frac{1}{p} \log \left(\frac{pe}{M-1}\right)\right) \right]^+,\\
	R_p^{\mathcal{X}^\prime}(D) &\leq -d_I(M-1) \log \left(\min\left\{ D, \frac{1}{M-1} \right\}\right).
    \end{align}
\end{theorem}

The proof is provided in Appendix \ref{app:pointwise}.
The content of Theorem \ref{thm:pointwise.l1} is that the rate distortion function is at least as great as the differential entropy of the regression function evaluated at the interpolation set, less a penalty term that involves the expected $L_p$ distortion and the cardinality of the interpolation set. The higher the cardinality of the interpolation dimension, the larger the rate-distortion function and the more nats are needed to describe the regression function on average. We emphasize that the lower bound holds for any interpolation set. If its cardinality is less than the interpolation dimension, one simply obtains a looser bound on $R_p^{\mathcal{X}^\prime}(D)$. However, the upper bound depends on the interpolation dimension, and in fact is invariant to the choice of interpolation set. In this sense, the interpolation dimension is a fundamental quantity, figuring prominently in both upper and lower bounds on the rate-distortion function.

\begin{corollary} Under the same assumptions as in Theorem~\ref{thm:pointwise.l1}, the rate-distortion function $R_p^{\mathcal{X}^\prime}(D)$ for $p \in\{1, 2, \infty\}$ is bounded from below by
	\begin{align}
    	 R_1^{\mathcal{X}^\prime}(D) &\geq \left[ h(W(S)) - d_*(M-1)\left(\log\left(\frac{2e}{M-1}D\right) \right)\right]^+ ,\\
	 R_2^{\mathcal{X}^\prime}(D) & \geq  \left[ h(W(S) - d_*(M-1)\log\left(\sqrt{\frac{2\pi e }{M-1}} D\right)\right]^+,\\
R_\infty^{\mathcal{X}^\prime}(D)& \geq    	\left[ h(W(S) - d_*(M-1)\log\left(2D \right)\right]^+.
\end{align}
\end{corollary}

\subsection{Bounds involving $\mathcal{L}_p(\delta)$}
The bounds presented in Theorem \ref{thm:pointwise.l1} are worst-case over a set $\mathcal{X}^\prime$ of test points, so a single poor-performing point 
forces a high value of the $\mathcal{X}^\prime$-Bayes risk. This worst-case performance is a useful metric, but we are also interested in knowing the average-case performance over the test points. To bound this performance, we carry out an analysis of the mutual information between $W$ and $\hat{W}$ averaged over an ensemble of interpolation sets. This requires additional machinery.
\begin{definition}\label{def:interpolation.map}
Let $\mathcal{V}$ be an index set, either countable or uncountable. We say that the function $\mathfrak{S}: \mathcal{V} \to \mathcal{X}^{d_*}$ is an {\em interpolation map} if: (1) every $\mathfrak{S}(v)$ is an interpolation set, and (2) $\mathfrak{S}(v) \cap \mathfrak{S}(v^\prime) = \emptyset$ for all $v \neq v^\prime \in \mathcal{V}$.
\end{definition}
In other words, an interpolation map defines a collection of disjoint interpolation sets, each having $d_*$ elements. The interpolation sets need not be {\em sufficient} interpolation sets, and the size $d_*$ of each interpolation set need not be the interpolation dimension $d_I$. Not every $x \in \mathcal{X}$ is found in an interpolation set $\mathfrak{S}(v)$, but we will see that an interpolation map that ``covers'' more of $\mathcal{X}$ is more useful for analysis.
\begin{definition}
	Let the {\em range} of $\mathfrak{S}(v)$ be
	\begin{equation*}
		\mathcal{W}(\mathfrak{S}) = \{x : \exists v \in \mathcal{V} \text{s.t.} x \in \mathfrak{S}(v)\}.
	\end{equation*}
	Then, let the {\em probability} of $\mathfrak{S}(v)$ be
	\begin{equation*}
		\gamma(\mathfrak{S}) = \int_{\mathcal{W}(\mathfrak{S})} p(x)dx,
	\end{equation*}
	where the integral becomes a sum if $\mathcal{X}$ is countable.
\end{definition}
We will find it convenient to work with interpolation maps that are {\em isotropic} with respect to the probability distribution $p(x)$.
\begin{definition}
	We say an interpolation map $\mathfrak{S}(v)$ is {\em isotropic} if, for every $v \in \mathcal{V}$ and $x,x^\prime \in \mathfrak{S}(v)$, $p(x) = p(x^\prime)$.
\end{definition}
In other words, an interpolation map is isotropic if the points in the interpolation set lie on level sets of the probability distribution $p(x)$. When one can define an isotropic interpolation map, one can bound the Bayes risk with expressions similar to those of Theorem \ref{thm:pointwise.l1}.

\begin{theorem}\label{thm:total.average}
	Suppose either that $\mathcal{X}$ is countable or that $\mathcal{X}$ is uncountable and the density $p(x)$ is Riemann integrable. Suppose also there exists an onto, sufficient interpolation set with cardinality $d_I$. Let $\mathfrak{S}$ be an isotropic interpolation map with dimension $d_*$ and probability $\gamma(\mathfrak{S})$. Then, the rate distortion function $R_p(D)$ for all $p\geq 1$ is bounded by
	\begin{align}\label{eqn:total.avg.lp}
R_p(D) & \geq	\left[ E_V[h(W(\mathfrak{S}(V)))] - d_* (M-1) \left( \log \left(\frac{D}{\gamma(\mathfrak{S})} \right)+  \log \left(2 \Gamma \left(1+\frac{1}{p} \right)\right) + \frac{1}{p} \log \left(\frac{pe}{M-1}\right)\right) \right]^+ , \\
R_p(D) & \leq	-d_I(M-1) \log \left(\min\left\{ D, \frac{1}{M-1} \right\}\right),
    \end{align}
    where the expectation is taken over the distribution $p(v) = p(\mathfrak{S}^{-1}(v))$, for $\mathfrak{S}^{-1}(v)$ denoting any point in the inverse image of $\mathfrak{S}(v)$. Furthermore, the lower bound on $R_\infty(D)$ holds for {\em any} interpolation map, regardless of whether or not it is isotropic.
\end{theorem}
      See Appendix \ref{app:total.average} for the proof.
The content of Theorem \ref{thm:total.average} is that, similar to the bounds on $R^{\mathcal{X}^\prime}_p$, the rate-distortion function depends on the entropy of the interpolation set and a penalty term associated with the permissible $\mathcal{L}_p$ risk. Here, however, the entropy is averaged over the interpolation map. This allows us to account for the average $L_p$ error over all of the points in the range of the interpolation map $\mathfrak{S}$, whereas the previous result accounted for the worst-case $L_p$ error over the subset $\mathcal{X}^\prime$.

\begin{corollary}
Under the same assumptions as in Theorem~\ref{thm:total.average}, the rate distortion function $R_p(D)$ for $p \in \{1,2, \infty\}$ is lower bounded by
	\begin{align}\label{eqn:total.avg.l1}
    	R_1(D) & \geq \left[ E_V[h(W(\mathfrak{S}(V)))] - d_*(M-1)\log\left(\frac{2e}{(M-1)\gamma(\mathfrak{S})}D\right) \right]^+  ,\\
    	R_2(D) &\geq \left[E_V[h(W(\mathfrak{S}(v)))] - d_*(M-1)\log\left(  \sqrt{\frac{2\pi e }{(M-1)\gamma(\mathfrak{S})}} D\right)\right]^+ ,\\
R_\infty(D) &\geq    	\left[ E_V[h(W(\mathfrak{S}(v)))] - d_*(M-1)\log\left(2D \right)\right]^+ . 
    \end{align}
\end{corollary}

% !TEX root = it.tex

\section{Sample Complexity Bounds}\label{sect:sample.complexity}
Combining Lemma \ref{lem:rd.to.bayes.risk} and Theorems \ref{thm:pointwise.l1} and \ref{thm:total.average}, one can compute lower bounds on the Bayes risk $\mathcal{L}_p$ and $\mathcal{L}_p^{\mathcal{X}^\prime}$, and turn them into sample complexity lower bounds. These bounds incorporate the interpolation dimension, the differential entropy of the posterior, and the mutual information between the training samples and the distribution index $\theta$.
To this end, one must evaluate two quantities: the mutual information $I(Z^n;\theta)$ between the training set and the parameterization $\theta$, and the differential entropy $h(W(S))$ of the regression function evaluated at points in the interpolation set, perhaps averaged over an interpolation map $\mathfrak{S}$. Under appropriate conditions, these quantities can be expressed in simpler terms, permitting the explicit computation of sample complexity bounds. We present the expressions for $I(Z^n;\theta)$ and $h(W(S))$, respectively, after which we derive sample complexity bounds.

\subsection{An Expression for $I(Z^n;\theta)$ for a Smooth Posterior}
Using results from universal source coding, we express $I(Z^n;\theta)$ in terms of the differential entropy of $\theta$ and the Fisher information matrix of the distribution family. Let $\alpha = s(\theta)$ be a {\em minimal sufficient statistic} of $\theta$, meaning both that we have that $p(x,y|\alpha,\theta) = p(x,y|\alpha)$ and that $\alpha$ is a function of any other sufficient statistic.

Then, let let $\mc{I}(\alpha)$ denote the Fisher information matrix with respect to $\alpha$:
\begin{equation*}
	\mc{I}(\alpha)_{i,j} \triangleq \! -E_{X,Y}\! \left[ \!\frac{\partial^2}{\partial \alpha_i \partial \alpha_j} \log \!\left(p(X,Y;\theta)\right)\!\right].
\label{eq:fisher}
\end{equation*}
The Fisher information roughly quantifies the amount of information, on the average, that each training sample conveys about $\theta$. Under appropriate regulatory conditions, we can make this notion precise and bound the mutual information in terms of the number of training samples $n$.
\begin{theorem}\label{thm:data.mi}
	Let the parametric family $p(x,y;\theta)$ have minimal sufficient statistic $\alpha \in \mathbb{R}^t$. Suppose that the Fisher information matrix $\mathcal{I}(\alpha)$ exists, is non-singular, and has finite determinant. Further suppose that the maximum-likelihood estimator of $\alpha$ from $Z^n$ is asymptotically efficient. That is, $(\hat{\alpha}(Z^n) - \alpha)\sqrt{n}$ converges to a normal distribution with zero mean and covariance matrix $\mathcal{I}^{-1}(\alpha)$. Then, the following expression holds
	\begin{equation}\label{eqn:data.mi}
		I(Z^n;\theta) = \frac{t}{2} \log \frac{n}{2\pi e} +  E_\alpha[\log |\mc{I}(\alpha)|^{\frac{1}{2}}]  + h(\alpha)  + o_n(1).
	\end{equation}
\end{theorem}
\begin{proof}
This follows from the celebrated redundancy-capacity theorem (see~\cite{Gallager-source-coding,Merhav_Feder_IT, ISIT11}). Averaging the bounds derived by Clarke and Barron~\cite{Clarke_Barron} over $q(\alpha)$ yields the result.
\end{proof}
The upshot is that the information conveyed by the training set grows as $(1/2)\log(n)$ times the effective dimension of the parameter space. Further constants are determined by the sensitivity of the distribution to the parameters, as quantified by the Fisher information matrix, and the prior uncertainty of the parameters, as expressed by $h(\alpha)$. The expression is intuitive in light of the assumption that the central limit theorem holds. The maximum-likelihood estimator of $\alpha$ approaches a Gaussian distribution in the limit of increasing $n$, and the resulting mutual information includes a term with the associated differential entropy. However, this result holds only for smooth distributions $p(x,y;\theta)$, for which the Fisher information matrix is also well-defined. In Section \ref{sec:zero-error}, we consider a case where the smoothness assumption does not hold, which changes the scaling on $I(Z^n;\theta)$.

\subsection{Expressions for $h(W(S)$}
The differential entropy $h(W(S))$ is an unusual quantity. To compute it, we must evaluate the {\em density of the posterior distribution $W(y|x;\theta)$}, evaluated at finitely many points, and take the expected logarithm. This ``density of a distribution'' will often be difficult to evaluate in closed form, and evaluating the expected logarithm will be more difficult still. Therefore, we cannot expect that a closed-form expression for $h(W(S))$ will always be available for problems of interest.

Nevertheless, we can develop intuition for $h(W(S))$. Consider an interpolation set $S$ with $|S| = d_I$. For example, for the binary Gaussian classifier, the interpolation set is a basis of $\mathbb{R}^d$ evaluated at $y=1$. Indeed, a common case is where $l=d$, and where the set of vectors $x_i$ form a basis of $\mathbb{R}^m.$ 

In this case, the differential entropy has the following expression.
\begin{theorem}\label{thm:differential.entropy}
      Define the random variables $N_{iy} = p(x_i,y|\theta)$, for every $1 \leq i \leq d_I$ and every $1 \leq y \leq M$. Then, under the preceding assumptions, the differential entropy of the regression function is
      \begin{equation}
       h(W(S)) = -\sum_{i=1}^{d_I} (M-1)E\left[\log\left(\frac{1+2S_i}{1+S_i} \right)\right] - \sum_{i=1}^{d_I}E\left[\log\left(1+ \frac{S_i}{(1+S_i)(1+2S_i)} \right)\right] + h(N),
      \end{equation}
      where $S_i = \sum_{y=1}^M N_{iy}$, and where $N$ is the matrix of the random variables $N_{iy}$, supposing that $h(N)$ exists.
\end{theorem}
The proof is given in Appendix \ref{app:differential.entropy}.
We make a few remarks about the preceding expression. First, it is a function of both the parametric family $p(x,y;\theta)$ and the prior $q(\theta)$. If the resulting distribution $p(N)$ is simple, then one can compute $h(W(S))$ relatively easily; otherwise, one may need to resort to numerical techniques. Whether or not it is simpler to estimate numerically the terms in the expression or to directly estimate $h(W(S))$ depends on the specific distributions in question. Second, as $M \to \infty$, for continuous distributions the first two terms converge with probability one on $d_I(M-1) \log(2)$. Finally, even if the expression is difficult to compute, it establishes scaling laws on $h(W(S))$, showing that it increases linearly in $d_I$ for parametric families satisfying the preceding conditions and corresponding to continuous distributions.

We also can bound the entropy from above.
\begin{theorem}\label{thm:max.entropy}
      For an interpolation set satisfying the conditions stated above, the posterior entropy satisfies
      \begin{equation}
            h(W(S)) \leq -d_I(M-1)\log\left(M-1\right).
      \end{equation}
\end{theorem}
      See Appendix \ref{app:differential.entropy} for the proof.
A fortiori, this expression bounds the expected entropy over an interpolation map $\mathfrak{S}(v)$. As before, the entropy grows linearly with the interpolation dimension. Further, here we see that the entropy decays at least as fast as $-M\log(M)$. While this bound is not necessarily tight, it is useful for rule-of-thumb estimation of the bounds when the true entropy is difficult to obtain. For example, in the multi-class Gaussian case, considered in Section \ref{sect:numerical.examples}, we find that substituting the preceding bound in place of the exact differential entropy only negligibly impacts the resulting bounds.

Finally, once an interpolation map $\mathfrak{S}(v)$ is identified, the expected entropy can be computed numerically via Monte Carlo methods, such as those presented in \cite{kraskov:PRE04}. As long as one can sample easily from $q(\theta)$, one can produce arbitrarily many samples of $W(y|x;\theta)$ with which to estimate the entropy, and one can further average over an appropriate interpolation map. We emphasize that this computation need only be carried out once for a given distribution model. Then, the risk bounds can be evaluated for any $n$.

\subsection{Sample Complexity Bounds}
From Theorem \ref{thm:data.mi} we can derive explicit sample complexity bounds for the $L_p$ Bayes risk. Substituting (\ref{eqn:data.mi}) into Lemma \ref{lem:rd.to.bayes.risk} and the lower bound~\eqref{eq:pointwise.lp} of Theorem \ref{thm:pointwise.l1}, we obtain a necessary condition
\begin{align}
\frac{t}{2} \log\left(\frac{n}{2\pi e}\right) \geq &-E_\alpha[\log |I(\alpha)^{-1}|] +(h(W(S)) - h(\alpha)) \nonumber \\
& -d_* (M-1) \left( \log \mc{L}_p^{\mc{X}^\prime} +  \log \left(2 \Gamma \left(1+\frac{1}{p} \right)\right) + \frac{1}{p} \log \left(\frac{pe}{M-1}\right)\right) + o(1) %\\
      %\log(n) &\geq \frac{1}{t}E_\alpha[\log |I(\alpha)^{-1}|] + \frac{2}{t}(E_X[h(\{W\}_{\mc{S}(X,y)})] - h(\alpha)) + \frac{2d_I}{t}\left(\log\left(\frac{1}{\bar{R}_\infty}\right)-1\right) + \log(2\pi e) + o(1),
\end{align}
for any interpolation set $S$, where $ \mc{L}_p^{\mc{X}^\prime}$ is the permissible $L_p$ loss on $\mc{X}^\prime$. Similar bounds hold for the Bayes risk $\mathcal{L}_p$ using~\eqref{eqn:total.avg.lp} in Theorem~\ref{thm:total.average}. We can describe intuitively the terms in this more complicated expression. The expression involving the Fisher information matrix describes how many samples are needed on average to learn the minimal sufficient statistic $\alpha$ and thus $\theta$. However, the objective is to learn $W$, and it is sufficient but perhaps not necessary to learn $\theta$ in order to learn $W$. The second term captures this notion with the difference between the entropies $h(W(S))$ and $h(\alpha)$; when $h(\alpha)$ is bigger, the term corrects the number of samples needed. The following term is the slack term associated with the $L_p$ Bayes risk, and the final $o(1)$ term arises from the approximation of $I(Z^n;\theta)$.

Further, if we impose regularity conditions on the eigenvalues of $\mathcal{I}(\alpha)$ and the differential entropies $h(\{W\}_S)$ and $h(\alpha)$, we derive the following bounds on the sample complexity.
\begin{proposition}\label{thm:sample.complexity.bounds}
      In addition to the conditions of Theorem \ref{thm:data.mi}, suppose that $E_\alpha[\log |I(\alpha)^{-1}|] \leq c_1$, and $h(W(S))- h(\alpha) \geq c_2$, for positive constants $c_1$ and $c_2$. Then,
      \begin{align}
           \mc{L}_p^{\mc{X}^\prime}  &\geq \frac{\exp\left(\frac{c_2 - c_1}{d_* (M-1)} \right)}{c_{3,p}}\left(\sqrt{\frac{2\pi e}{n}}\right)^{t/(d_* (M-1))}(1+o(1)) ,
      \end{align}
      where $c_{3,p} = 2 \Gamma \left(1+\frac{1}{p} \right) \left(\frac{pe}{M-1}\right)^{ \frac{1}{p}}$.
\end{proposition}
\begin{proof}
      Algebraic manipulation on Theorem \ref{thm:data.mi}.
\end{proof}
For the special case $t=d_* (M-1)$, in which the dimensionality of the minimal sufficient statistic $\alpha$ is equal to the interpolation dimension times the alphabet size, we obtain the order-wise rule $\mc{L}_p^{\mc{X}^\prime}  = \Omega(\sqrt{1/n})$. This agrees with the known scaling on the sample complexity in this case, which one can derive from the PAC framework. That is, in this case, it is both necessary and sufficient for the number of training samples to grow quadratically in the required precision. We emphasize that the $\Omega(1/\sqrt{n})$ scaling is particular both to the choice of the $L_p$ loss and the assumption of smoothness. Indeed, for the classification loss $L_c$ or the KL divergence $L_{\mathrm{KL}}$, there exist learning problems for which one can derive PAC error bounds that decay as $O(1/n)$ \cite{bartlett:AS05,massart:AS06,mendelson:JACM15}. Furthermore, there exist classification problems for which one can derive $O(1/n)$ error bounds even for the $L_p$ loss. An important, albeit extreme, example is the so-called {\em error free} learning problem, in which there exists a classifier that perfectly classifies test samples. In Section \ref{sect:numerical.examples} we derive rate-distortion bounds for a simple error-free setting; the resulting bounds imply a $\Omega(1/n)$ scaling on the $L_p$ loss as expected.

\subsection{Achievability}
A natural question is whether matching upper bounds on the Bayes risk hold. The random coding and joint typicality arguments that prove achievability bounds for rate distortion do not apply to the supervised learning setting. Joint typicality arguments depend on repeated i.i.d. draws from the source distribution. The analogous situation in supervised learning would be to encode multiple i.i.d. posteriors, each associated with a different draw from $q(\theta)$. Of course, we consider only a single draw from $p(\theta)$. Furthermore, random coding arguments presuppose design control over the source code, which would correspond to design control over the distribution of the labeled data. We do not have such control, so these arguments do not apply.

As stated in Proposition \ref{thm:sample.complexity.bounds}, under mild assumptions one can derive lower bounds on the $L_p$ error that scale as $\sqrt{1/n}$. One can further apply the PAC framework to derive upper bounds that scale as $\sqrt{1/n}$. In this case, the rate-distortion bounds are order optimum.
However, the constants and higher order terms of the rate-distortion and PAC bounds differ significantly, primarily because PAC bounds are distribution agnostic while Bayes risk bounds take the prior distribution $q(\theta)$ into account.

Finally, one can derive {\em asymptotic} achievability bounds on the $L_p$ Bayes risk via analysis of the plug-in estimate of the posterior. Suppose the posterior is Lipschitz continuous in the $L_p$ norm with respect to $\theta$ and the conditions of Theorem \ref{thm:data.mi} hold. Then central limit theorem implies that the estimation error of $\hat{\theta}$ has variance scaling as $1/n$ and depending on the Fisher information matrix. By the Lipschitz assumption, the Bayes risk scales as $1/\sqrt{n}$, with constants and higher-order terms depending on the Fisher information matrix and the Lipschitz constant.
% !TEX root = it.tex

\section{Case Studies}\label{sect:numerical.examples}
\subsection{Categorical Distribution}
We consider first a comparatively simple learning problem: learning a discrete probability distribution from samples drawn i.i.d. from the distribution. As mentioned in the introduction, this problem has been studied extensively. We model this problem in our framework by posing a multi-class learning problem where the alphabet $\mathcal{X}$ is trivial. That is, let $\mathcal{X}= \{0\}$, let $M$ be the number of values the random variable may take, and let the family of distributions be indexed by $\Lambda = \Delta_{M-1}$, the $M-1$-dimensional unit simplex. That is, the ``posterior'' distribution $p(y|x,\theta)$ is simply the distribution of the categorical random variable $Y$, with probabilities indicated by $\theta$:
\begin{equation}
      p(y|x,\theta) = \theta_y.
\end{equation}
To compute the Bayes risk, we suppose a Dirichlet prior over $\theta$, with hyperparameter $\gamma \in \mathbb{R}^M_+$. In particular, if each $\gamma_i=1$, then $\theta$ is uniformly distributed over the unit simplex. Define $\gamma_0 := \sum_{i=1}^L \gamma_i$.

Trivially, there is a single sufficient interpolation set $S = \{0\}$, and the interpolation dimension of the categorical distribution is $d_I=1$. Therefore an interpolation {\em map} is trivial, with $\mathcal{V} = \{0\}$ and $\mathfrak{S}(0) = \{0\}$. 
Since $\mathcal{X}$ is a singleton, $R_p(D) = R_p^{\mathcal{X}}(D)$ and hence we use $R_p(D)$ to refer to both.
The entropy of the Dirichlet distribution is well-known, so we can evaluate the differential entropy of the posterior $h(W(S))$ at the unique interpolation set:
\begin{equation}\label{eqn:categorical.differential.entropy}
	h(W(S)) = h(\gamma_1,\dots,\gamma_{M-1}) = \log (B(\gamma)) - (M-\gamma_0)\psi(\gamma_0) - \sum_{i=1}^M(\gamma_i-1)\psi(\gamma_i),
\end{equation}
where $B(\gamma)$ is the multivariate Beta function and $\psi(z)$ is the digamma function. This allows us to compute the $L_p$ rate-distortion functions.
\begin{theorem}\label{thm:categorical.rd}
	For a categorical distribution, the rate distortion function $R_p(D)$ is bounded by
	\begin{align}
R_p(D)	 \geq & \left[ \log (B(\gamma)) - (M-\gamma_0)\psi(\gamma_0) - \sum_{i=1}^M(\gamma_i-1)\psi(\gamma_i) \right. \nonumber \\&\left.-  (M-1) \left( \log D +  \log  \left(2 \Gamma \left(1+\frac{1}{p} \right)\right) + \frac{1}{p} \log \left(\frac{pe}{M-1}\right)\right) \right]^+,\\
R_p(D)  \leq	& -(M-1) \log \left(\min\left\{ D, \frac{1}{M-1} \right\}\right).
	\end{align}
\end{theorem}
\begin{IEEEproof}
	Evaluating the bounds in Theorem \ref{thm:pointwise.l1} for the differential entropy in (\ref{eqn:categorical.differential.entropy}) yields the result.
\end{IEEEproof}
In particular, we are interested in the $L_1$ lower  bounds given as
\begin{corollary}
For a categorical distribution, the rate distortion function $R_1(D)$ is lower bounded by
	\begin{equation}
 R_1(D) \geq	\left[\log (B(\gamma)) - (M-\gamma_0)\psi(\gamma_0) - \sum_{i=1}^M(\gamma_i-1)\psi(\gamma_i) - (M-1)\left(\log\left( \frac{2e}{M-1}D \right) \right)\right]^+ .
	\end{equation}
\end{corollary}

The lower bound on rate-distortion functions depend on the cardinality $M$ of the alphabet for $Y$ and the prior $p(\theta)$. For the uniform case, i.e. $\gamma_i=1$, the differential entropy of the posterior simplifies to $h(\gamma_1,\dots,\gamma_{M-1}) = \log (B(\gamma)) = -\log((M-1)!)$. In this case, the upper and lower bound differ only by this term, the magnitude of which grows in $M$ without bound.

Using Lemma \ref{lem:rd.to.bayes.risk}, we can translate the rate-distortion bounds to lower bounds on the Bayes risk. This requires the evaluation of the mutual information $I(Z^n;\theta)$. This mutual information is difficult to compute in closed form, so we apply the approximation from Theorem \ref{thm:data.mi}.
\begin{lemma}\label{lem:categorical.mi}
	For learning a categorical distribution the mutual information $I(Z^n;\theta)$ is
	\begin{equation}
      I(Z^n;\theta) = \frac{M-1}{2}\log\left(\frac{n}{2\pi e}\right) + \frac{M-1}{2}\psi(\gamma_0) - \frac{1}{2}\sum_{i=1}^{M-1}\psi(\gamma_i) + h(\theta_1,\dots,\theta_{M-1}) + o_n(1).
\end{equation}
\end{lemma}
The proof appears in Appendix~\ref{app:proof-lemmas}.
These results imply the following bound on the $L_p$ Bayes risk.
\begin{theorem}\label{thm:categorical.risk.bounds}
	For learning a categorical distribution, the $\mathcal{L}_p$ Bayes risk is bounded below as
	\begin{align}\label{eqn:categorical.l1.bound}
		\mathcal{L}_1 &\geq (M-1)\sqrt{\frac{\pi}{2 e n}}\exp\left(\frac{1}{2(M-1)}\sum_{i=1}^{M-1}\psi(\gamma_i)- \psi(\gamma_0)\right)(1+o_n(1)) \\
		\mathcal{L}_\infty &\geq \sqrt{\frac{\pi e}{2n}}\exp\left(\frac{1}{2(M-1)}\sum_{i=1}^{M-1}\psi(\gamma_i)- \psi(\gamma_0)\right)(1+o_n(1)) \\
		\mathcal{L}_2 &\geq \sqrt{\frac{M-1}{n}}\exp\left(\frac{1}{(M-1)}\sum_{i=1}^{M-1}\psi(\gamma_i)- \psi(\gamma_0)\right)(1+o_n(1)).
	\end{align}
\end{theorem}
\begin{IEEEproof}
	Manipulation on Theorem \ref{thm:categorical.rd} and Lemma \ref{lem:categorical.mi}.
\end{IEEEproof}

In other words, the $L_p$ Bayes risk decays no faster than $\sqrt{1/n}$, with a constant that depends on the prior distribution, $p$, and an unspecified $o_n(1)$ term. Furthermore, we can derive a lower bound on the minimax $L_1$ error by choosing $\gamma = (\kappa,\kappa,\dots,\kappa,0)$, i.e. all but the $M$th element is a constant. Using the fact that $\lim_{x \to \infty} \exp(\psi(x)) = x - 1/2$, we obtain
\begin{align}
	\mathcal{L}_1 &\geq \lim_{\kappa \to \infty} (M-1)\sqrt{\frac{\pi}{2 e n}}\exp\left(\frac{1}{2}(\psi(\kappa)- \psi((M-1)\kappa))\right)(1+o_n(1)) \\
	&= \sqrt{\frac{\pi (M-1)}{2 e n}}(1+o_n(1)). \label{eqn:categorical.rd.minimax.constant}
\end{align}

For comparison, \cite{kamath:15} gives upper and lower bounds on the $L_1$ minimax error, where the lower bounds result from bounds on the Bayes risk. In particular, \cite{kamath:15} shows that
\begin{multline}\label{eqn:categorical.kamath.bounds}
	\sqrt{\frac{2(M-1)}{\pi n}}\left(1- \frac{M}{2(M-1)\kappa}\right) - \frac{4M^{1/2}(M-1)^{1/4}}{n^{3/4}} - \frac{M(1-M\kappa)}{n + M\kappa} \leq
	\mathcal{L}_1  \leq  \\
	\sqrt{\frac{2(M-1)}{\pi n}} + \frac{4M^{1/2}(M-1)^{1/4}}{n^{3/4}},
\end{multline}
for symmetric priors $\gamma_i = \kappa$ for $\gamma \geq 1$. For large $\kappa,$ the lower and upper bounds agree suggesting the sample complexity of $\sqrt{\frac{2(M-1)}{\pi n}} $.
 Furthermore, the constant in (\ref{eqn:categorical.rd.minimax.constant}) is quite close to the the constant in the first term of the bounds of (\ref{eqn:categorical.kamath.bounds}). Numerically, $\sqrt{\frac{2}{\pi}} - \sqrt{\frac{\pi}{2 e}} \approx 0.0377$. It is somewhat remarkable that our proposed method leads to a bound that is so close to the actual minimax bound given that it is not adapted and optimized for the problem at hand. An open question is whether further optimization over the prior $q(\theta)$ can eliminate this gap.

We point out a few advantages of the bounds proven in Theorem \ref{thm:sample.complexity.bounds}. First of all, the lower bounds of Theorem \ref{thm:categorical.risk.bounds} hold for any prior, while the lower bound in (\ref{eqn:categorical.kamath.bounds}) is only applicable for $\gamma \geq 1$. Furthermore, the negative terms in (\ref{eqn:categorical.kamath.bounds}) lead to a negative bound on the $L_1$ risk for many values of $n$. 
On the other hand, the rate-distortion bounds are asymptotic in that there is an unspecified multiplicative constant of $(1+o_n(1))$ owing to the approximation to the mutual information $I(Z^n;\theta)$. Therefore, we emphasize that the objective of this exercise is not to derive bounds that outperform or are tighter than those in the literature, but to illustrate the effectiveness of the proposed framework. With relatively little effort, the framework provides risk bounds that are tight order-wise and nearly tight to within a small multiplicative constant.

\subsection{Multinomial Classifier}
	Here we consider binary classification under the {\em multinomial} model. For this model, we will derive rate-distortion and Bayes risk bounds on the $\mathcal{L}_P^\mathcal{X}$ error, i.e. the {\em worst-case} Bayes risk over test points $X$.

	Under the multinomial model, the alphabet is $\mathcal{X} = \mathbb{Z}_+^d$, we suppose $M=2$, and the observation $X$ has the class-conditional distribution
	\begin{equation}
		p(x | y, \theta) = \frac{k!}{\prod_{i=1}^d x_i!} \prod_{i=1}\theta_{iy}^{x_i},
	\end{equation}
	where $\theta \in \mathbb{R}^{d \times 2}$ $y \in \{1,2\}$, $k \in \mathbb{Z}$ is a known constant, and the distribution is supported over the set of non-negative vectors $x \in \mathcal{X}$ such that $\sum_{i=1}^d x_i = k$.

	The multinomial distribution is a generalization of the binomial distribution. Given $k$ variates drawn from a $d$-ary categorical distribution, $X_i$ counts the number of times the variate has the value $i$. The multinomial model is ubiquitous in text and document classification, where word or part-of-speech counts are used to identify document types. Essentially, each class is modeled by a different categorical distribution, and the task of the classifier is to determine from which categorical distribution the $k$ test samples were drawn.

	Let $\theta^1 \in \mathbb{R}^d$ and $\theta^2 \in \mathbb{R}^d$ denote the categorical distribution parameters for each class-conditional distribution. Suppose also uniform priors, i.e. $\mathrm{Pr}(Y=1) = \mathrm{Pr}(Y=2) = 1/2$. The resulting posterior is
	\begin{equation}\label{eqn:multinomial.general.posterior}
		p(y=1 | x,\theta) = \frac{1}{1+\prod_{i=1}^d (\theta_{i}^2/\theta_{i}^1)^{x_i}}.
	\end{equation}
	
	Similar to the previous subsection, we suppose a Dirichlet prior, which is the conjugate prior. Specifically, let $\theta^2 \sim \mathrm{Dir}(\gamma)$, where $\gamma \in \mathbb{R}_+^{d}$ and where $\gamma_0 := \sum_i \gamma_i$. We also choose $\theta^1 = (\theta^2-1)$, and we let $\theta := \theta^2$. This choice allows for a more straightforward derivation of the rate-distortion function and the accompanying sample complexity bounds. With these choices, the posterior has the form
	\begin{equation}\label{eqn:multinomial.restricted.posterior}
		p(y=1 | x,\theta) = \frac{1}{1+\prod_{i=1}^d R_i^{x_i}},
	\end{equation}
	where $R_i = \theta_i/(1-\theta_i)$. Because $R_i \in [0,\infty)$, this parameterization does not cost any generality in terms of posterior functions; any posterior of the form of (\ref{eqn:multinomial.general.posterior}) can be put in the form of (\ref{eqn:multinomial.restricted.posterior}). Thus, the distribution over $\theta$ induces a prior over all multinomial Bayes classifiers, and the Bayes risk bounds for any prior are lower bounds on the minimax risk.

	To compute the rate-distortion function, we find an interpolation set for the posterior $W$. It is straightforward to see that $S = \{k \mathbf{e}_1, k\mathbf{e}_2, \dots, k\mathbf{e}_{d-1}\}$ is a sufficient interpolation set. 
	Evaluating the posterior at each point $x = k\mathbf{e}_i$ yields
	\begin{equation}
		W(k \mathbf{e}_i) = \frac{1}{1+R_i^k},
	\end{equation}
	which allows one to recover the ratio $R_i$ and therefore the parameter $\theta_i$; since each categorical distribution lies on the unit simplex, one can recover the entire parameter vector $\theta$. It is also straightforward to see that any smaller interpolation set is not sufficient; hence $d_I = d-1$.

	Using this interpolation set, we bound the rate-distortion function over the {\em worst-case} data point $x \in \mathcal{X}$.
	To do so, we first need to compute $h(W(S))$, which requires the following lemma.
	\begin{lemma}\label{lem:multinomial.entropy}
		For the scalar function $g(r) = \frac{1}{1+r^k}$, let $V := g(R)$ for a random variable $R$ such that $h(R)$ exists. Then,
		\begin{equation}
			h(V) \geq h(R) + \log(k) - \frac{2}{k}\log(2) - 2E[[\log(R)]^+] + (k-1)E[\log(R)].
		\end{equation}
	\end{lemma}
See Appendix~\ref{app:proof-lemmas} for the proof.
		From this, we can bound $h(W(S))$.
	\begin{lemma}\label{lem:multinomial-diff-entropy}
		For the multinomial classification problem described and the interpolation set $S = \{k\mathbf{e}_1,\dots,k\mathbf{e}_{d-1}\}$, the differential entropy $h(W(S))$ is upper bounded by
		\begin{multline}
			h(W(S)) \geq (d-1)\left(\log(k) - \frac{2}{k}\log(2)\right) + \\
			\sum_{i=1}^{d-1}\left(  \log(B(\gamma_i,\gamma_0-\gamma_i)) + (\gamma_0+\gamma_i + 2 - k)\psi(\gamma_0-\gamma_i) - (\gamma_0 - 2)\psi(\gamma_0) + (k-\gamma_i)\psi(\gamma_i)\right).
		\end{multline}
	\end{lemma}
The proof appears in Appendix~\ref{app:proof-lemmas}.
	From this result, we can state bounds on the rate-distortion function for the multinomial classifier.
	\begin{theorem}\label{thm:multinomial.rd}
		For binary multinomial classification, the rate-distortion functions $\mathcal{R}_1^\mathcal{X}(D)$ is bounded by
	\begin{align}
 R^\mathcal{X}_p(D)	& \geq h(W(S))  - (d-1) \left( \log D +  \log \left(2 \Gamma \left(1+\frac{1}{p} \right)\right) + \frac{1}{p} \log \left(pe\right)\right), \\
  R^\mathcal{X}_p(D)	& \leq -(d-1) \log(\min\{D, 1\}).
	\end{align}
	\end{theorem}
	Note that this bound is for the worst-case distortion over all test points $x \in \mathcal{X}$. In order to bound the distortion averaged over $x$, we would need to construct an interpolation map $\mathfrak{S}$ that spans a sufficiently large portion of $\mathcal{X}$. Because the alphabet $\mathcal{X}$ involves the constraint that $\sum_i x_i = k$, it is difficult to construct such an interpolation map. Therefore, for now we let the bound on $R^\mathcal{X}_p(D)$ suffice, and we relegate average-case bounds to future work.

	Finally, from the rate-distortion bounds, we can derive Bayes risk bounds. Again it is difficult to compute $I(Z^n;\theta)$ directly, so we resort to the approximation of Theorem \ref{thm:data.mi}.
	\begin{lemma}\label{lem:multinomial.mi}
		For the multinomial classification problem described, the mutual information $I(Z^n;\theta)$ is
		\begin{multline}
			I(Z^n;\theta) = \frac{d-1}{2}\log\left(\frac{n}{2\pi e}\right) + \log (B(\gamma)) - (d-\gamma_0)\psi(\gamma_0) - \sum_{i=1}^d(\gamma_i-1)\psi(\gamma_i) + \\
			\frac{d-1}{2}\log(k/2) - \frac{1}{2}\sum_{i=1}^{d-1}(\psi(\gamma_i) + \psi(\gamma_0-\gamma_i) - 2\psi(\gamma_0)) + o_n(1).
		\end{multline}
	\end{lemma}
	Putting these results together yields a bound on the $\mathcal{X}$-Bayes risk. For brevity, we state only the bounds for $\mathcal{L}_1^\mathcal{X}$; it is easy to infer the other bounds by analogy.
	\begin{theorem}
		For the multinomial binary classification problem, the $\mathcal{X}$-Bayes risk for any learning rule is bounded below by
		\begin{multline}
			\mathcal{L}_1^{\mathcal{X}} \geq k 2^{-\frac{2+k}{k}}\sqrt{\frac{2\pi e}{n}} \times 
			\exp\left(-\frac{1}{d-1}B(\gamma) + \left(1-\gamma_0 + \frac{d-\gamma_0}{d-1}\right)\psi(\gamma_0) + \frac{\gamma_d-1}{d-1}\psi(\gamma_d) + \right. \\{}
			\left.  \frac{1}{d-1}\sum_{i=1}^{d-1}[(k-1/2)\psi(\gamma_i) + (\gamma_0 + \gamma_i + 2 - k)\psi(\gamma_0 - \gamma_i) \right).% \right).
		\end{multline}
	\end{theorem}
	\begin{IEEEproof}
		Manipulation on Theorem \ref{thm:categorical.rd} and Lemma \ref{lem:categorical.mi}.
	\end{IEEEproof}
The resulting bounds are somewhat involved, but one can glean an intuition from them. As $k$ increases, the family of classifiers becomes more rich, and the Bayes risk increases. However, the error decays as $\sqrt{1/n}$. As mentioned earlier in the paper, the VC dimension bounds imply that the risk cannot decay more {\em slowly} that $\sqrt{1/n}$. This lower bound establishes that this scaling is order-optimum for the $L_p$ Bayes risk. However, we cannot prove tightness of constants as we could in the categorical case.

\subsection{Binary Gaussian Classifier}\label{sect:binary.gaussian}
Next, we consider a binary Gaussian setting. Let $M=2$ and $p(y) = 1/2, y \in \{1,2\}$, and let $\Lambda = \mathbb{R}^d$ parameterize the data distributions. The class-conditional densities are Gaussian with antipodal means:
\begin{align}
\label{eq:px-y}
      p(x|y=1; \theta) &= \mathcal{N}(\theta,\sigma^2 \mathbf{I}) \nonumber\\
      p(x|y=2; \theta) &= \mathcal{N}(-\theta,\sigma^2 \mathbf{I}),
\end{align}
where $\sigma^2 > 0$ is the known variance. We choose the prior $q(\theta) = \mathcal{N}(0,(1/d)\mathbf{I})$. The MAP classifier for this problem is a hyperplane passing through the origin and normal to $\theta$. For this model, the regression function is
\begin{equation}
      W(y=1|x, \theta) = \frac{1}{1+\exp(-2/\sigma^2 x^T\theta)}.
\end{equation}
This model is the continuous-valued analogue to the multinomial case in that the posterior has the form of a logistic function. 
For this model, however, we can derive closed-form bounds on the rate-distortion function $R_p(D)$ averaged over test points $X$. To see this, we first observe that any orthogonal basis of $\mathbf{R}^d$ is a sufficient interpolation set for $W$. For a basis set $S = \{x_1,\dots,x_d\}$, each function evaluation $W(y=1 | x=x_i, \theta)$ allows one to recover the inner product $x_i^T \theta$; $d$ independent inner products allows one to recover $\theta$ and interpolate the entire regression function. Furthermore, choosing an orthogonal basis guarantees that the elements of $W(S)$ are statistically independent. Next, we define an interpolation {\em map}. Let the index set be the positive orthant of $\mathbb{R}^d$:
\begin{equation*}
	\mathcal{V} = \{x \in \mathbb{R}^d : x_i < 0 \}.
\end{equation*}
Then, let the interpolation map be
\begin{equation*}
	\mathfrak{S}(v) = \{\text{an orthogonal basis } S: v \in S, \norm{x} = \norm{v}, \forall x \in S\}.
\end{equation*}
It is straightforward to verify that one can choose the basis such that the sets $\mathfrak{S}(v)$ are disjoint for different choices of $v \in \mathcal{V}$. It is also straightforward to verify that the range of the interpolation map is $\mathcal{W}(\mathfrak{S}) = \mathbb{R} \setminus \{0\}$ and that the probability of the interpolation map is $\gamma(\mathfrak{S}) = 1$.

In order to bound the rate-distortion function, we compute the expected entropy of the regression function averaged over the interpolation map.
\begin{lemma}\label{lem:binary.entropy}
      The expected differential entropy of the posterior evaluated at the interpolation map is bounded by
      \begin{multline*}
      	E_V[h(W(\mathfrak{S}(V)))] \geq \frac{d}{2}\psi(d/2) + \frac{d}{2}\log\left(\frac{16\pi(1/(d\sigma^2)+1)}{d\sigma^2}  \right) - \\ \frac{d\Gamma((d+1)/2)}{\Gamma(d/2)}\sqrt{\frac{4(1/(d\sigma^2)+1)}{\pi d \sigma^2}} - \frac{3d}{2} - 2d\log(2),
      \end{multline*}
     where $\Gamma(\cdot)$ is the Gamma function, and $\psi(\cdot)$ is the digamma function, and where $\nu(d,\sigma^2)$ is defined as the lower bound divided by $d$. Let $\nu(d,\sigma^2)$ denote the preceding bound divided by $d$.
\end{lemma}
Then, we have a bound on the rate-distortion function. For brevity we state the result only for $\mathcal{L}_1$.
\begin{theorem}\label{thm:binary.gaussian.rd}
	For binary Gaussian classification, the rate-distortion function is bounded by
	\begin{multline}
		\left[ \frac{d}{2}\psi(d/2) + \frac{d}{2}\log\left(\frac{16\pi(1/(d\sigma^2)+1)}{d\sigma^2}  \right) - \right.\\ \left. \frac{d\Gamma((d+1)/2)}{\Gamma(d/2)}\sqrt{\frac{4(1/(d\sigma^2)+1)}{\pi d \sigma^2}} - \frac{3d}{2} - 2d\log(2) - d\log(2eD)\right]^+ \leq R_1(D) \leq - d\log(\min\{D, 1\}).
	\end{multline}
\end{theorem}
\begin{IEEEproof}
	Evaluation of the bounds of Theorem \ref{thm:total.average} for the differential entropy of Lemma \ref{lem:binary.entropy}.
\end{IEEEproof}
Furthermore, we can bound the Bayes risk. Again we state the results for the $L_1$ risk only for brevity.
\begin{theorem}
	For the binary classification problem, the $L_1$ Bayes risk is bounded by
	\begin{equation}
     \mathcal{L}_1 \geq \sqrt{\frac{\sigma^2 d}{\sigma^2 d + n}}\exp(\nu(d,\sigma^2) - 1).
	\end{equation}
\end{theorem}
\begin{IEEEproof}
	First, we bound $I(Z^n;\theta)$. First, one can verify that $\theta$ is itself a minimal sufficient statistic, and $\mathcal{I}(\theta) = 1/\sigma^2 \mathbf{I}$. It is also immediate that $h(\theta) = d/2\log(2\pi e/d)$. Combining these facts with Theorem \ref{thm:data.mi}, we obtain
\begin{equation}
      I(Z^n;\theta) = \frac{d}{2} \log \left(\frac{n}{d\sigma^2}\right)  + o(1). \label{eq:asymptote}
\end{equation}
However, in this case we can evaluate the mutual information in closed form. Because the prior $q(\theta)$ is  Gaussian, the posterior is not only asymptotically Gaussian but also Gaussian for any $n$. Let $T_i = \theta + N_i$, where $N_i \sim \mc{N}(0, \sigma^2 \mathbf{I})$. Simple calculation shows that
\begin{align}
I(Z^n; \theta) &= I(T^n; \theta)\nonumber\\
& = h(\theta) - h(\theta|T^n)\nonumber\\
&=\frac{d}{2} \log \left(1 + \frac{n}{d\sigma^2}\right).\nonumber
\end{align}
The discrepancy between the estimated and exact mutual information is negligible unless $n \ll d\sigma^2$. Applying this to the bound in Theorem \ref{thm:binary.gaussian.rd} yields the result.
\end{IEEEproof}
In this case we obtain a scaling law of $\sqrt{1/n}$. One can again invoke the PAC bounds based on the VC dimension to see that this scaling also upper bounds the Bayes risk. Therefore this bound is order-optimum, although its tightness under more strict definitions remains an open question.

% !TEX root = it.tex

\subsection{Zero-error Classification}
\label{sec:zero-error}
Finally, we present an example of a zero-error binary classification problem. In this experiment, the samples are generated such that they are truly linearly separable, where it is well known that the required sample complexity drastically changes to merely $O(1/n)$ as opposed to the usual $O(1/\sqrt{n})$ in the case of noisy samples (see~\cite{ehrenfeucht1989general} and the references therein).

We consider a parametric distribution with a single parameter (i.e., a one-dimensional parameter vector). Let $\Theta = [0,1]$, and let $\theta \sim \mathcal{U}[0,1]$. We let $\mc{X} = [0,1]$ and $\mc{Y} = \{-1,1\}$. We let $P_\theta(z)$ be defined as follows.
\begin{equation}
P_\theta(x, y) = \left\{ \begin{array}{ll}
1 & \text{if $(x - \theta)y>0$}\\
0 & \text{if $(x - \theta)y<0$}\\
\end{array}.
\right.
\label{eq:01-joint}
\end{equation}
Note that the joint distribution defined in~\eqref{eq:01-joint} does not satisfy Clarke and Barron's smoothness condition as it is clearly not continuous. Hence, we will need to calculate the mutual information directly, and we need to calculate $P(\theta|Z^n)$.
We show that in fact the mutual information follows a different scaling in this case leading to a different sample complexity scaling law as expected.

To compute $I( Z^n; \theta)$, we first focus on $P_\theta(z^n)$. 
We have
\begin{equation}
P_\theta(z^n) = \left\{ \begin{array}{ll}
1 & \text{if $(x_i - \theta)y_i>0$ for all $i \in [n]$}\\
0 & \text{otherwise}\\
\end{array}
\right.
= \left\{ \begin{array}{ll}
1 & \text{if $\theta_l<\theta<\theta_r$}\\
0 & \text{otherwise}\\
\end{array},
\right.
\end{equation}
where 
\begin{equation}
\theta_l := \max_{y_i=-1} \{x_i\} \quad  \text{and} \quad \theta_r := \min_{y_i = 1} \{x_i\}.
\label{eq:theta-lr}
\end{equation} 
Note that when no sample $x_i$ exists such that $y_i = -1$, we define $\theta_l =0$, and similarly when so sample $x_i$ exists such that $y_i= 1$, we define $\theta_r = 1$. By using the Bayes' rule, we get
\begin{equation}
P(\theta|z^n) 
= \left\{ \begin{array}{ll}
\frac{1}{\theta_r - \theta_l} & \text{if $\theta_l<\theta<\theta_r$}\\
0 & \text{otherwise}\\
\end{array}.
\right.
\end{equation}
Thus, we have
\begin{equation}
I(\theta; Z^n) = h(\theta) - h(\theta|Z^n) = -h(\theta|Z^n).
\end{equation}
We have
\begin{align}
I(Z^n; \theta) &= \int P_\theta(z^n) P(\theta) \log P(\theta|z^n) dz^n d\theta\\
& = -  \int P_\theta(z^n) \log (\theta_r - \theta_l) \mathbf{1}_{\theta_l <\theta<\theta_r} (\theta)dz^n d\theta.
\end{align}
Hence, integrating $\theta$ over $[0,1]$, we have
\begin{equation}
I(Z^n; \theta)  =  - \sum_{i=0}^{n} \int  (x_{(i+1)} - x_{(i)}) \log (x_{(i+1)} - x_{(i)}) dz^n,
\end{equation}
where $\{x_{(i)}\}_{i=1}^n$ are the order statistics of $x^n$, and $x_{(0)}:= 0$ and $x_{(n+1)}:=1$.
Note that due to the linearity of expectation and symmetry the above can be written as 
\begin{align}
I( Z^n; \theta)  &=  - (n+1) \int  x_{(1)}  \log x_{(1)}  dz^n\\
&=  - (n+1) n \int  x_{(1)} (1-x_{(1)})^{n-1} \log x_{(1)}  dx_{(1)}\\
&=  H_{n+1} - 1,
\label{eq:zero-error-mi1}
\end{align}
where $H_n$ is the n'th harmonic number defined as 
\begin{equation}
H_n := \sum_{i \in [n]} \frac{1}{i}.
\end{equation}
Hence, it is evident that asymptotically as $n \to \infty$,
\begin{equation}
I(Z^n;\theta) = \log n + O(1).
\label{eq:zero-error-mi2}
\end{equation}
In other words, $I(Z^n;\theta)$ scales as $\log n$. This is in contrast to the rest of the case studies where Clarke and Barron would apply, and the mutual information between the samples and the unknown parameter would scale as $\frac{1}{2} \log n$. 

Rather than constructing an interpolation map, in this case, it is more straightforward to bound the rate-distortion function directly. First, note that
\begin{align}
 \mathcal{L}_p(\delta)^p  &= E_{X,Z^n,\theta}\left[\sum_{y=1}^2 |W(y|x;\theta) - \hat{W}(y|x) |^p\right]\\
 &= 2 E_{X,Z^n,\theta}\left[|W(1|x;\theta) - \hat{W}(1|x) |^p\right].
\end{align}
It is straightforward to show that the procedure that minimizes  $\mathcal{L}_p(\delta)$ is $\hat{W}(y|x)  = W(y|x;\hat{\theta})$ where $\hat{\theta} = \hat{\theta}(Z^n)$ is the maximum-likelihood estimate of $\theta$.
Hence, without loss of generality, we only focus on such strategies
\begin{align}
 \mathcal{L}_p(\delta)^p   &= 2 E_{X,Z^n,\theta}\left[|W(1|x;\theta) - W(1|x; \hat{\theta}) |^p\right]\\
  &= 2 E_{Z^n,\theta}\left[|\theta - \hat{\theta}|^p\right]\\
  & = 2 E_{u} \left[|u|^p \right],
\end{align}
where $u = \theta - \hat{\theta}$. Hence, the rate-distortion function is calculated by
\begin{equation}
R_p(D) \geq \left[- \max_{p(u): 2 E_{u} \left[|u|^p \right] <D^p} h(u) \right]^+.
\end{equation}
The above is maximized by 
\begin{equation}
p_u(u ) = \frac{\lambda^{1/p}}{ 2 \Gamma \left(\frac{p+1}{p}\right)} e^{- \lambda |u|^p},
\end{equation}
where $\lambda = \frac{2}{p D^p}$. The value of the R is hence given by
\begin{align}
R_p(D) & = \left[ - \log \left(2\Gamma \left( \frac{p+1}{p}\right)\right) - \frac{1}{p} \log (pe)  - \log D\right]^+.
\label{eq:zero-error-R_p}
\end{align}
Hence, combining~\eqref{eq:zero-error-mi1},~\eqref{eq:zero-error-mi2}, and~\eqref{eq:zero-error-R_p}, the sample complexity lower bound that we arrive at in this example has a fundamentally different scaling of $O(1/n)$ as opposed to the $O(1/\sqrt{n})$ in all of the previous examples. Note that this bound is order-wise tight as it is consistent with matching upper bounds~\cite{ehrenfeucht1989general}.

Let us further enumerate this bound and analyze the tightness of the sample complexity bounds using this framework. Considering the $L_1$ Bayes risk, we consider the following estimator:
$$\hat{\theta} = \hat{\theta}(z^n) = \frac{1}{2} (\theta_r+ \theta_l),$$
where $\theta_l$ and $\theta_r$ are defined in~\eqref{eq:theta-lr}. Thus, due to the uniform prior, 
\begin{equation}
E[|\theta - \hat{\theta}|  ] = \frac{1}{4}E\left[\theta_r - \theta_l \right] = \frac{1}{4(n+1)}.
\end{equation}
Hence, to satisfy an $L_1$ risk of $\mc{L}_1 = 2 E[|\theta - \hat{\theta}|  ] < \frac{1}{2}$, it would suffice to have: 
\begin{equation}
n \leq \frac{1}{2\mc{L}_1} - 1. 
\end{equation}
On the other hand, the rate-distortion bounds suggest the following lower bound,
$
H_{n+1} - 1 \geq - \log \mc{L}_1 - 1  ,
$
which in turn is equivalent to the following
\begin{equation}
n \geq e^{-\gamma} \frac{1}{2 \mc{L}_1} - 1 + o(1) ,
\end{equation}
where $\gamma$ is the Euler's constant, and $e^{-\gamma} \approx 0.56$. As suggested by the upper bound given by this scheme, the lower bound is tight within a multiplicative factor of $2$.

\section{Conclusion}
\label{sect:conclusion}
We have presented a rate-distortion framework for analyzing the learnability of Bayes classifiers in supervised learning. Treating the regression function as a random object, we derived bounds on its rate-distortion function under average $L_p$ loss. We showed that the rate distortion function is bounded above and below by expressions involving the {\em interpolation dimension}, a new quantity that characterizes in a sample-theoretic fashion the complexity of a parametric family. In addition to characterizing the amount of information needed to describe a Bayes classifier up to a specified $L_p$ loss, we showed that the rate-distortion function permits the derivation on lower bounds on the $L_p$ Bayes risk in terms of the number of samples. We evaluated these bounds for several statistical models, in some cases showing that the bounds are nearly tight.

An important future application of this work is in distributed learning. The rate-distortion bounds characterize both how much information is required to learn a classifier and how much information is required to {\em describe} a classifier that one has already learned. As a result, we expect that they will prove useful in characterizing how much information a network of classifiers, connected by rate-limited links, will need to exchange in order to learn collaboratively a Bayes classifier from their distributed samples.

\appendices
% !TEX root = it.tex

%\section{Proof of Lemma}

\section{Proof of Theorem \ref{thm:pointwise.l1}}\label{app:pointwise}
The proof of Theorem \ref{thm:pointwise.l1} involves the maximization of the entropy of the posterior estimation error $U(y|x)~:=~W(y|x;\theta)~-~\hat{W}(y|x)$ subject to the $\mathcal{L}_p^{\mathcal{X}^\prime}$ constraint, for which we prove a lemma.
\begin{lemma}\label{lem:pointwise.l1.error.entropy}
       Let $S \subset \mathcal{X}^\prime$ be an interpolation set with cardinality $|S| = d_*$. Then, if $p(\hat{W} | W)$ satisfies the constraint $\mathcal{L}_p^{\mathcal{X}^\prime} \leq D$, the differential entropy $U(S) = W(S)-\hat{W}(S)$ satisfies
       \begin{equation*}
        h(U(S)) \leq d_* (M-1) \left( \log D +  \log \left(2 \Gamma \left(1+\frac{1}{p} \right)\right) + \frac{1}{p} \log \left(\frac{pe}{M-1}\right)\right).
        \end{equation*}
\end{lemma}
\begin{IEEEproof}
       We start by bounding
       \begin{equation}\label{eqn:sum.of.entropies}
              h(U(S)) \leq \sum_{x \in S}\sum_{y=1}^{M-1} h(U(y|x)),
       \end{equation}
       which holds with equality if and only if each $U(y|x)$ is independent for all $(x,y)$. By hypothesis, $U(y|x)$ has $L_p$ norm no greater than $D$ for every $x \in \mathcal{X}^\prime$ and in expectation over the remaining random variables, or
       \begin{equation}
              \sum_{y=1}^M E[|U(y|x)|^p] \leq D^p.
       \label{eqn:lp.verbose}
       \end{equation}
       This leads to the optimization problem
       \begin{equation}
       \begin{aligned}
                  &\underset{p(U)}{\text{maximize}} \quad & & \sum_{y=1}^{M-1} h(U(y|x)), \\
                  &\text{subject to} \quad & & \sum_{y=1}^{M} E[|U(y|x)|^p] \leq D^p.
%                  & & & |U(y|x)|\leq 1,~ \forall y \in \{ 1, \ldots, M\}. 
       \end{aligned}
       \end{equation}
       The preceding formulation suggests that $U(y|x)$ does not impact the objective function for $y=M$. Therefore, it is clearly optimal to set $U(y=M|x) = 0$ with probability one, and the resulting optimization problem is
       \begin{equation}
       \begin{aligned}
          &\underset{p(U)}{\text{maximize}} \quad & & \sum_{y=1}^{M-1} h(U(y|x)), \\
          &\text{subject to} \quad & & \sum_{y=1}^{M-1} E[|U(y|x)|^p] \leq D^p.
%          & & & |U(y|x)|\leq 1, ~\forall y \in \{ 1, \ldots, M-1\}.
       \end{aligned}
       \end{equation}
This is a convex program, and writing down the Lagrangian, we get that the optimizer is such that $U(y|x)$ is independent and identically distributed for all $y \in \{1, \ldots, M-1\}$. Further, 
\begin{equation}
p_{U(y|x)}(u) \propto e^{-\lambda |u|^p},
\label{eq:max-dist}
\end{equation}
where $\lambda$ has to be set such that the expectation constraint,  $E[|U(y|x)|^p] \leq D^p / (M-1)$, is satisfied for each $y \in \{1, \ldots, M-1\}$. The claim of the lemma  is the established by the calculation of $h(U(y|x)$ for this distribution.
\end{IEEEproof}
Now we are ready to prove Theorem \ref{thm:pointwise.l1}.
\begin{IEEEproof}[Proof of Theorem \ref{thm:pointwise.l1}]
\noindent {\bf Lower bound:} By the data-processing inequality,
\begin{align*}
      R_p(D) &\geq \inf_{p(U)} h(W(S)) - h(U(S)),
\end{align*}
where recall that $U(S) = W(S) - \hat{W}(S)$. Applying Lemma \ref{lem:pointwise.l1.error.entropy} yields the lower bound.

\noindent {\bf Upper bound: }  By the definition of $R_p^{\mathcal{X}^\prime}(D)$, for any sufficient interpolation set we have
\begin{align*}
  R_p(D) &= \inf_{p(\hat{W}|W)} I(W;\hat{W}) \\
  &=  \inf_{p(\hat{W}|W)} I(W(S);\hat{W}(S)) \\
  &= \inf_{p(\hat{W}|W)} h(\hat{W}) - h(\hat{W}(S) | W(S)) \\
  &= \inf_{p(U)} h(\hat{W}) - h(W(S) + U(S) | W(S)) \\
  &\leq  \inf_{p(U)} - h(U(S) | W(S)),
\end{align*}
where the inequality follows because $\hat{W}$ is a posterior function, so each element is a member of $[0,1]$, thus the joint differential entropy is bounded above by zero. Finally, choose $U(y|x)$ for all $y \in \{1,\ldots, M-1\}$ to be jointly independent of $W$, and choose the distribution of its elements to be i.i.d. with the following distribution: 
For a given $W(y|x)$, if $D<1$, each $U(y|x)$ follows
\begin{equation}
U(y|x) \sim \mc{U}\left(-\frac{D W(y|x)}{M-1}, \frac{D(1-W(y|x))}{M-1}\right),
\end{equation}
and if $D>1$,
\begin{equation}
U(y|x) \sim \mc{U}\left(-\frac{W(y|x)}{M-1}, \frac{(1-W(y|x))}{M-1}\right).
\end{equation}
The constraint is set up to ensure that $|U(M|x)|\leq1$ and $0\leq \hat{W}(y|x) \leq 1$, which in turn will ensure the expectation constraint. Note that the actual value of $W(y|x)$ does not play a role in the differential entropy of $U$ as it is merely a translation, and hence the value of $h(U(y|x))$ does not depend on $W(y|x)$. The upper bound is achieved by evaluating the differential entropy of this choice of $U(y|x)$.
\end{IEEEproof}

% !TEX root = it.tex

\section{Proof of Theorem \ref{thm:total.average}}\label{app:total.average}
\noindent {\bf Lower Bound:}
By the data-processing inequality and the fact the supremum dominates any average,
\begin{align*}
      R_p(D) &\geq \sup_{v \in \mathcal{V}} \inf_{p(U)} h(W(\mathfrak{S}(v))) - h(U(\mathfrak{S}(v))) \\
      %&\geq \sup_{x,y} \inf_{p(U)}  h(\{W\}_{\mc{S}(x,y)}) - h(\{U\}_{S(x,y)}) \\
      &\geq \inf_{p(U)} E_V[h(W(\mathfrak{S}(V))) - h(U(\mathfrak{S}(V)))],
\end{align*}
where again $U = W - \hat{W}$, and where the expectation is over any distribution $p(v)$ over the set $\mathcal{V}$; that is, $p(v)$ is a mass function if $\mathcal{V}$ is countable and a density function if $\mathcal{V}$ is uncountable. We will specify $p(v)$ presently.

For a given $p$, each $U(y|x)$ is subject to the expectation constraint $E |U(y|x)|^p$, and the objective is to maximize the differential entropy term $h(U(\mathfrak{S}(v)))$. As was shown in the proof of Theorem~\ref{thm:pointwise.l1}, the optimizing distribution is given in~\eqref{eq:max-dist}, where we have to set $D = D(y|x)$ to be a function of $(x,y)$ for each $U(y|x)$.
Then, we have
\begin{equation}\label{eqn:average.gap1}
      R_p(D) \geq E_V[h(W(\mathfrak{S}(V)))] - \sup_{D(y|x)} E_V\left[\sum_{x \in \mathfrak{S}(V)}\sum_{y=1}^{M-1}(\log(D(y|x))+C_p)\right],
\end{equation}
where $C_p$ is defined as
\begin{equation}
C_p = \log \left(2 \Gamma \left(1+\frac{1}{p} \right)\right) + \frac{1}{p} \log \left(\frac{pe}{M-1}\right).
\end{equation}
It remains to take the supremum of the sum of logarithm terms while respecting the distortion constraint. To simplify the optimization problem, we suppose that $\mathcal{X}$ and $\mathcal{V}$ are countable; thus $p(x)$ and $p(v)$ are probability mass functions. For uncountable $\mathcal{X}$, the result follows by taking the limit of countable partitions. Note that the term $D(y=M|x)$ does not show up in the RHS of (\ref{eqn:average.gap1}), so, similar to the proof of Lemma \ref{lem:pointwise.l1.error.entropy}, it is optimum to set $D(y=M|x) = 0$. Further, it follows from the concavity of $\log(\cdot)$ that it is optimal that $D(y|x) = D_0$ is a constant for all $x \in \mathcal{W}(\mathfrak{S})$ and all $y \in \{1, \ldots, M-1\}$.
 For every $x \notin \mathcal{W}(\mathfrak{S})$, it is clearly optimal to set $D(y|x) = 0$. Then, the constraint on $\mathcal{L}_1$ becomes
\begin{align}
	\sum_{y=1}^{M-1}\sum_{x \in \mathcal{W}(\mathfrak{S})} p(x) D(y|x) &= D \\
	\implies (M-1) \gamma(\mathfrak{S}) D(y|x) &= D \\
	\implies D(y|x) &= \frac{D}{(M-1)\gamma(\mathfrak{S})},
\end{align}
recalling that $\gamma(\mathfrak{S})$ is the probability of the range of $\mathfrak{S}$. Substituting this into (\ref{eqn:average.gap1}) yields
\begin{align}
	R_p(D) &\geq E_V[h(W(\mathfrak{S}(V)))] - d_*(M-1)(\log(D/\gamma(\mathfrak{S}))+C_p), %\\
	%&\geq \sup_{v} h(W(\mathfrak{S}(v)))] - d_*(M-1)(\log(2\epsilon/(M\gamma(\mathfrak{S}))+1).
\end{align}
as was to be shown.

\noindent {\bf Upper Bound:} The upper bound follows from the same argument as the bound on $R_p^{\mathcal{X}^\prime}(D)$.

% !TEX root = it.tex

\section{Proofs of Theorems \ref{thm:differential.entropy} and \ref{thm:max.entropy}}\label{app:differential.entropy}
\begin{IEEEproof}[Proof of Theorem \ref{thm:differential.entropy}]
By Bayes' rule, the posterior at each pair in the interpolation set is
\begin{equation*}
      W(x_i|y_i,\theta) = \frac{N_{iy_i}}{\sum_{y=1}^M N_{iy}}.
\end{equation*}
In order to compute $h(W(S))$, we need to compute the density $p(W(S))$. A standard result \cite{gelman} is that the density is
\begin{equation}
      p_W(W(S)) = |J|p_N(f(W(S))),
\end{equation}
where $f$ is the one-to-one function mapping the samples $W(S)$ to the terms $N_{iy}$, and were $J$ is the Jacobian matrix of $f^{-1}$. $W_i = (W_{i1}, \dots W_{iM-1})$ and $N_i = (N_{i1},\dots,N_{iM-1})$. Because of the normalization constraint, the random variables $W_{iy}$ are overdetermined by the random variables $N_{iy}$. Therefore, without loss of generality we can take $N_{iM}=1$ for every $i$. Then, it is straightforward to show that the mapping between $W_i$ and $N_i$ is
\begin{equation}
       N_i = W_i\left(1+\frac{S_i}{1+S_i} \right) \triangleq g(W_i),
\end{equation}
where $S_i = \sum_{y=1}^{M-1} W_{iy}$. Taking derivatives, the determinant of the Jacobian, denoted $|J_i|$, is
\begin{equation}
       |J_i| = \left(\frac{1+2S_i}{1+S_i} \right)^{M-1}\left(1+ \frac{S_i}{(1+S_i)(1+2S_i)} \right).
\end{equation}
Now, let $W$ be the matrix of all vectors $W_i$ and $N$ be the matrix of all vectors $N_i$. Clearly the Jacobian is block diagonal, thus the Jacobian of the entire mapping, denoted $|J|$, is
\begin{equation}
       |J| = \prod_{i=1}^{k}\left(\frac{1+2S_i}{1+S_i} \right)^{M-1}\left(1+ \frac{S_i}{(1+S_i)(1+2S_i)} \right).
\end{equation}
The density of the matrix $W$ is thus $|J|p(N)$, where $p(N)$ is the density of the matrix $N$. The resulting entropy is therefore
\begin{multline}
       h(W(S)) = -E[\log p(W)] = -\sum_{i=1}^k (M-1)E\left[\log\left(\frac{1+2S_i}{1+S_i} \right)\right] - \\ \sum_{i=1}^kE\left[\log\left(1+ \frac{S_i}{(1+S_i)(1+2S_i)} \right)\right] + h(N).
\end{multline}
\end{IEEEproof}

Next, we prove Theorem \ref{thm:max.entropy}.
\begin{IEEEproof}[Proof of Theorem \ref{thm:max.entropy}]
      The proof follows the same structure as the proof of Lemma \ref{lem:pointwise.l1.error.entropy}. The entropy $h(W(S))$ is bounded above by the sum of the individual entropies. Furthermore, as $W(y|x,\theta)$ is normalized and non-negative, the sum of the regression function over $y$ for a fixed $x$ must be equal to one. The entropy is maximized by letting each random variable be uniformly distributed across $[0,1/(M-1)]$, and the result follows.
\end{IEEEproof}

%\section{Proof of Theorem \ref{thm:binary.entropy}}\label{app:binary.entropy}

\section{Proofs of Lemmas}
\label{app:proof-lemmas}

\begin{IEEEproof}[Proof of Lemma~\ref{lem:rd.to.bayes.risk}]
Suppose there is a learning rule $\delta(Z^n)$ that satisfies the Bayes risk constraint $\mathcal{L}_p^{\mathcal{X}^\prime} \leq D$ or $\mathcal{L}_p\leq D$, respectively, and let $p(\theta,Z^n,W,\hat{W})$ be the joint probability distribution on the distribution index, training set, true posterior, and learned posterior. Because $\hat{W}$ is a function of the training set $Z^n$, we have the Markov chain: $W \to \theta \to Z^n \to \hat{W}$. Therefore, repeated applications of the data-processing inequality yield
\begin{equation*}
    I(Z^n;\theta) \geq I(W;\hat{W}),
\end{equation*}
where the latter mutual information is computed according to $p(W,\hat{W}) = p(\hat{W}|W)p(W)$, which is obtained by marginalizing the joint distribution. Furthermore, we can take the infimum over {\em all} distributions $p(\hat{W}|W)$ that satisify the Bayes risk constraint:
\begin{equation*}
    I(Z^n;\theta) \geq \inf_{p(\hat{W} | W)} I(W;\hat{W}).
\end{equation*}
The RHS of the preceding is the definition of the rate-distortion function for the appropriate Bayes risk.
\end{IEEEproof}

	\begin{IEEEproof}[Proof of Lemma~\ref{lem:categorical.mi}]
	Observe that while $\theta$ overdetermines the distribution and is not a minimal sufficient statistic, $\alpha = (\theta_1,\dots,\theta_{M-1})$ is minimal and sufficient. Straightforward calculation demonstrates that the Fisher information matrix for $\alpha$ is
\begin{equation}
 I(\alpha) = \mathrm{diag}(1/\theta_1,\dots,1/\theta_{M-1}).
\end{equation}
Therefore,
\begin{align}
      \frac{1}{2}E[\log|I(\alpha)|] &= -\frac{1}{2}\sum_{i=1}^{M-1} E[\log(\theta_i)] \\
      &= \frac{M-1}{2} \psi(\gamma_0) - \frac{1}{2}\sum_{i=1}^{M-1} \psi(\gamma_i),
\end{align}
where the latter equality holds because the marginal components $\theta_i$ of a Dirichlet distribution follow a Beta distribution, i.e. $\theta_i \sim \mathrm{Beta}(\gamma_i,\gamma_0-\gamma_i)$, and the expected logarithm of a Beta random variable is
\begin{equation*}
	E[\log (\theta_i) ] = \psi(\gamma_i) - \psi(\gamma_0).
\end{equation*}
Substituting the preceding into Theorem \ref{thm:data.mi} yields the result.
\end{IEEEproof}

\begin{IEEEproof}[Proof of Lemma~\ref{lem:multinomial.entropy}]
		It is straightforward to verify that $g$ is invertible with $g^{-1}(v) = \left(\frac{1-v}{v}\right)^{1/k}$. By the Jacobian method, the density of $V$ is
		\begin{equation}
			f_V(v) = f_R(g^{-1}(v)) \left| \frac{d}{dv}g^{-1}(v) \right| = f_R(g^{-1}(v)) \frac{(1/v-1)^{1/k}}{k(1-v)v},
		\end{equation}
		and the differential entropy is bounded by
		\begin{align*}
			h(V) &= -E[\log f_R(g^{-1}(V)) ] - \frac{1}{k}E[\log((1-V)/V)] +  E[\log(V(1-V)] + \log(k) \\
			&= h(R) + \log(k) + \frac{k+1}{k}E[\log(V)] + \frac{k-1}{k}E[\log(1-V)] \\ %- (1+k)E[\log(R)] + \log(k).
			&= h(R) + \log(k) + \frac{k+1}{k}E\left[ \log\left(\frac{1}{1+R^k}\right) \right] + \frac{k-1}{k}E\left[ \log\left(\frac{R^k}{1+R^k}\right)\right] \\
			&= h(R) + \log(k) - \frac{2}{k}E[\log(1+R^k)] + (k-1)E[\log(R)] \\
			&\geq h(R) + \log(k) - \frac{2}{k}E[\log(2\max\{1,R^k\})] + (k-1)E[\log(R)] \\
			&= h(R) + \log(k) - \frac{2}{k}\log(2) - 2E[[\log(R)]^+] + (k-1)E[\log(R)].
		\end{align*}
	\end{IEEEproof}

	\begin{IEEEproof}[Proof of Lemma~\ref{lem:multinomial-diff-entropy}]
		We have that $h(W(S)) = \sum_{i=1}^{d-1} h(W(k\mathbf{e}_i)$. By Lemma \ref{lem:multinomial.entropy}, we need to compute $h(R_i)$. Since $R_i = \theta_i/(1-\theta_i)$, it has a beta distribution of the second kind, with density
	\begin{equation}
		f_{R_i}(r) = \frac{r^{\gamma_i-1}(1+r)^{\gamma_0}}{B(\gamma_i,\gamma_0-\gamma_i)}.
	\end{equation}
	It therefore has differential entropy
	\begin{align*}
		h(R_i) &= \log(B(\gamma_i,\gamma_0-\gamma_i)) - (\gamma_i-1)E[\log(R_i)] - \gamma_0 E[\log(1+R_i)] \\
		&= \log(B(\gamma_i,\gamma_0-\gamma_i)) - (\gamma_i-1)E[\log(\theta_i)] + (\gamma_i-1)E[\log(1-\theta_i)] + \gamma_0E[\log(1-\theta_i)] \\
		&= \log(B(\gamma_i,\gamma_0-\gamma_i)) + (\gamma_0 + \gamma_i - 1)E[\log(1-\theta_i)] - (\gamma_i - 1)E[\log(\theta_i)] \\
		&= \log(B(\gamma_i,\gamma_0-\gamma_i)) + (\gamma_0 + \gamma_i - 1)(\psi(\gamma_0-\gamma_i) - \psi(\gamma_0)) - (\gamma_i - 1)(\psi(\gamma_i) - \psi(\gamma_0)),
	\end{align*}
	where the final equalities are due to the expression of the logarithm of a beta-distributed random variable. Similarly, we have that
	\begin{equation*}
		E[\log (R_i)] = E[\log(\theta_i)] - E[\log(1-\theta_i)] = (\psi(\gamma_i) - \psi(\gamma_0)) - (\psi(\gamma_0 - \psi(\gamma_i) - \gamma_0) = \psi(\gamma_i) - \psi(\gamma_0 - \gamma_i).
	\end{equation*}
	It remains to calculate $E[[\log(R_i)]^+]$:
	\begin{align*}
		E[[\log(R)]^+] &= E[\max\{0, \log(\theta_i) - \log(1-\theta_i) \}] \\
		&\geq E[\max\{0, -\log(1-\theta_i) \}] \\
		&= -E[\log(1-\theta_i)] \\
		&= \psi(\gamma_0) - \psi(\gamma_0 - \gamma_i).
	\end{align*}
	Invoking Lemma \ref{lem:multinomial.entropy} and combining terms yields the result
	\end{IEEEproof}
	
	\begin{IEEEproof}[Proof of Lemma~\ref{lem:multinomial.mi}]
		First, we observe that $\alpha = (\theta_1,\dots,\theta_{d-1})$ is a minimal sufficient statistic for the joint distribution. Then, it is straightforward to show that the Fisher information matrix is
		\begin{equation}
			\mathcal{I}(\alpha) = \mathrm{diag}\left(\frac{k}{2\theta_1(1-\theta_1)},\dots, \frac{k}{2\theta_{d-1}(1-\theta_{d-1})}\right).
		\end{equation}
		By Theorem \ref{thm:data.mi}, the mutual information is
		\begin{align*}
			I(Z^n;\theta) &= \frac{d-1}{2}\log\left(\frac{n}{2\pi e}\right) + h(\alpha) + \frac{1}{2}E[\log |\mathcal{I}(\alpha)|] + o_n(1) \\
			&= \frac{d-1}{2}\log\left(\frac{n}{2\pi e}\right) + \log (B(\gamma)) - (d-\gamma_0)\psi(\gamma_0) - \sum_{i=1}^d(\gamma_i-1)\psi(\gamma_i) + \\
			&\quad\quad \frac{d-1}{2}\log(k/2) - \frac{1}{2}\sum_{i=1}^{d-1}(\psi(\gamma_i) + \psi(\gamma_0-\gamma_i) - 2\psi(\gamma_0)) + o_n(1),
		\end{align*}
		which is the desired result.
	\end{IEEEproof}

  \begin{IEEEproof}[Proof of Lemma \ref{lem:binary.entropy}]
      For the interpolation set $\mathcal{S}(x,y)$, let $c=\norm{x}$, and let $W_i = W(s_i(x,y))$ denote the posterior evaluated at the $i$th element of the set determined by the interpolation map. Without loss of generality, suppose $y=1$. Further, let $Z_i = (2x_i^T\theta)/\sigma^2$. Straightforward computation shows that
      \begin{equation}
            Z_i \sim \mathcal{N}\left(0,\frac{4c^2}{\sigma^4 d}\right),
      \end{equation}
      and
      \begin{equation}
            W_i = \frac{1}{1+\exp(-Z_i)}.
      \end{equation}
      Our first objective is to find the density $p(W_i)$. Using the Jacobian formula,
      \begin{equation}
            p(W_i) = J(f)\cdot p_{N_i}(f(W_i)),
      \end{equation}
      where
      \begin{equation}
            f(W_i) = \log\left(\frac{W_i}{1-W_i} \right)
      \end{equation}
      is the mapping from $W_i$ to $N_i$, and $J(f)$ is the Jacobian of $f$, which is equal to
      \begin{equation}
            J(f) = \frac{\partial}{\partial W_i} f(W_i) = \frac{1}{W_i(1-W_i)}.
      \end{equation}
      Therefore,
      \begin{align}
            p(W_i) &= J(f) \cdot \mathcal{N}\left(f(W_i),\frac{4c^2}{\sigma^4 d}\right) \\
            &= \frac{1}{W_i(1-W_i)\sqrt{\frac{8\pi c^2}{d\sigma^4}}}\exp\left(-\frac{\sigma^2 d}{8c^2} \log^2\left(\frac{W_i}{1-W_i}\right)\right).
      \end{align}
      Next, the differential entropy is
      \begin{align}
            h(W_i) &= -E[\log(p(W_i))] \\
            &= \frac{1}{2}\log\left(\frac{8\pi c^2}{d\sigma^4}\right) + E[\log(W_i(1-W_i))] + E\left[\frac{\sigma^2 d}{8c^2}\log^2\left(\frac{W_i}{1-W_i}\right) \right].
      \end{align}
      Observe that
      \begin{equation}
            W_i(1-W_i) = \frac{\exp(-Z_i)}{(1+\exp(-Z_i))^2},
      \end{equation}
      and
      \begin{equation}
            \frac{W_i}{1-W_i} = \exp(Z_i).
      \end{equation}
      Therefore,
      \begin{align}
            h(W_i) &= \frac{1}{2}\log\left(\frac{8\pi c^2}{d\sigma^4}\right) - E[Z_i] + \frac{\sigma^4 d}{8c^2}E[Z_i^2] - 2E[\log(1+\exp(-Z_i))]  \\
            &= \frac{1}{2}\log\left(\frac{8\pi c^2}{d\sigma^4}\right) + \frac{1}{2} - 2E[\log(1+\exp(-Z_i))] \\
            &\geq \frac{1}{2}\log\left(\frac{8\pi c^2}{d\sigma^4}\right) + \frac{1}{2} - 2\log(2) - E[[Z_i]^+]\label{eq:ineq-logexp}\\
            & = \frac{1}{2}\log\left(\frac{8\pi c^2}{d\sigma^4}\right) - \frac{3}{2} -  2\log(2) - \sqrt{\frac{2  c^2}{\pi d  \sigma^4}} \label{eqn:folded.gaussian}
      \end{align}
     where~\eqref{eq:ineq-logexp} follows from the fact that $1+e^{-x} < 2e^{-x}$ for any $x<0$ and $1+e^{-x} \leq 2$ for any $x\geq 0$. Furthermore, (\ref{eqn:folded.gaussian}) follows from the fact that $E[[Z_i]^+]$ is exactly half the expectation of the ``folded'' Gaussian, which is well known. This bound is tight to within a constant gap of $\log(2)$.

     Our final step is to take the expectation over $V$, which has a white Gaussian distribution with per-element variance $1/d+\sigma^2$. Define the random variables $C = \norm{V}$ and $K = C^2/(1/d+\sigma^2)$, which yields
     \begin{equation}
      E_V[h(W_i)] \geq E\left[\frac{1}{2}\log\left(\frac{8\pi(1/(d\sigma^2)+1)K}{d \sigma^2} \right) - \sqrt{\frac{2(1/(d\sigma^2)+1)K}{\pi d \sigma^2}}\right] - \frac{3}{2} - 2\log(2).
     \end{equation}
     By definition, $K \sim \chi^2(d)$, so in order to evaluate the preceding expectation we need to compute the mean of a $\chi$-distributed random variable and the expected logarithm of a $\chi^2$ random variable. These quantities are well-known, and applying them yields
     \begin{multline}
      E_V[h(W_i)] \geq \frac{1}{2}\psi(d/2) + \frac{1}{2}\log\left(\frac{16\pi(1/(d\sigma^2)+1)}{d\sigma^2}  \right) - \\ \frac{\Gamma((d+1)/2)}{\Gamma(d/2)}\sqrt{\frac{4(1/(d\sigma^2)+1)}{\pi d \sigma^2}} - \frac{3}{2} - 2\log(2),
     \end{multline}
     where $\Gamma(\cdot)$ is the Gamma function, and $\psi(\cdot)$ is the digamma function.
\end{IEEEproof}

\bibliographystyle{IEEEtran}
%\small
\bibliography{bibliography,references}

\end{document}